\documentclass[aps,pra,twocolumn,amsmath,amssymb,nofootinbib,showpacs,superscriptaddress]{revtex4-1}
\usepackage[english]{babel}
\usepackage{latexsym}
\usepackage{graphics}
\usepackage{graphicx}
\usepackage{epsfig}
\usepackage{color}
\usepackage{bm}
\usepackage{amsmath}
\usepackage{amssymb}
\usepackage{amsthm}
\usepackage{dcolumn}
\usepackage{bm}
\usepackage{float}
\usepackage{braket}
\usepackage{hyperref}
\usepackage{color}
\usepackage{epstopdf}
\usepackage{cleveref}
\usepackage{mathrsfs}
\usepackage[svgnames]{xcolor}
\usepackage{mathtext}
\hypersetup{hidelinks,colorlinks=true,allcolors=DarkBlue}

\DeclareMathOperator{\Tr}{Tr} 

\newtheorem{theorem}{Theorem}

\theoremstyle{remark}
\newtheorem{remark}{Remark}

\begin{document}

\preprint{APS/123-QED}

\title{Security of the decoy state method for quantum key distribution}

\author{A.S. Trushechkin}
\affiliation{Steklov Mathematical Institute, Russian Academy of Sciences, Moscow 119991, Russian Federation}
\affiliation{Department of Mathematics and NTI Center for Quantum Communications, National University of Science and Technology MISiS, Moscow 119049, Russia}

\author{E.O. Kiktenko}
\affiliation{Steklov Mathematical Institute, Russian Academy of Sciences, Moscow 119991, Russian Federation}
\affiliation{Department of Mathematics and NTI Center for Quantum Communications, National University of Science and Technology MISiS, Moscow 119049, Russia}
\affiliation{Russian Quantum Center, Skolkovo, Moscow 143025, Russia}
\affiliation{Moscow Institute of Physics and Technology, Dolgoprudny, Moscow Region 141700, Russia}

\author{D.A. Kronberg}
\affiliation{Steklov Mathematical Institute, Russian Academy of Sciences, Moscow 119991, Russian Federation}
\affiliation{Russian Quantum Center, Skolkovo, Moscow 143025, Russia}
\affiliation{Moscow Institute of Physics and Technology, Dolgoprudny, Moscow Region 141700, Russia}

\author{A.K. Fedorov}
\affiliation{Russian Quantum Center, Skolkovo, Moscow 143025, Russia}
\affiliation{Moscow Institute of Physics and Technology, Dolgoprudny, Moscow Region 141700, Russia}

\date{\today}
\begin{abstract}
Quantum cryptography or, more precisely, quantum key distribution (QKD), 
is one of the advanced areas in the field of quantum technologies. 
The confidentiality of keys distributed with the use of QKD protocols is guaranteed by the fundamental laws of quantum mechanics. 
This paper is devoted to the decoy state method, a countermeasure against vulnerabilities caused by the use of coherent states of light for QKD protocols whose security is proved under the assumption of single-photon states. 
We give a formal security proof of the decoy state method against all possible attacks. 
We compare two widely known attacks on multiphoton pulses: photon-number splitting and beam splitting. 
Finally, we discuss the equivalence of polarization and phase coding.
\end{abstract}
\maketitle 

\section{Introduction}

Information protection is one of the key demands of modern society.
In most cases, information security is ensured by using cryptographic techniques, such as encryption.\,Encryption is commonly understood as the transformation of information that is needed to be secured (plaintext) into an encrypted message (ciphertext) with the use of a certain algorithm~\cite{Kuzmin2002}.
At the same time, to realize encryption, the legitimate parties of the communication need a so-called cryptographic key, which is a secret parameter (usually a binary string of a certain length),
which determines the choice of a specific transformation of information when performing encryption.
The key distribution problem is one of the most important in cryptography~\cite{Kuzmin2002, Schneier1996}.
For example, Ref.~\cite{Schneier1996} emphasizes: 
``Keys are as valuable as all the messages they encrypt, 
since knowledge of the key gives knowledge of all the messages. 
For encryption systems that span the world, the key distribution problem can be a daunting task.''

Several approaches to the distribution of cryptographic keys can be used.
First, the keys can be delivered using trusted couriers.
The main disadvantage of this method is the presence of the human factor.
Furthermore, with the increasing amount of transmitted data keys every year, physical transfer is becoming more difficult.
An alternative approach is public-key cryptography.
It is based on the use of so-called one-way functions, i.e., functions that are easy to compute but for which it is difficult to find an argument for a given function value.
Examples include the Diffie-Hellman and RSA (abbreviation from the surnames Rivest, Shamir, and Adleman) algorithms (which were developed for encrypting messages, but are also used for key distribution), 
which use the complexity of solving discrete logarithm and integer factorization problems, respectively.
The majority of data transmitted on the Internet is protected with the use of public-key algorithms, which are included in the HTTPS (HyperText Transfer Protocol Secure) protocol.

Today, efficient ({of polynomial complexity}) 
classical algorithms for solving factorization and discrete logarithm problems are unknown, but there is the efficient quantum Shor algorithm~\cite{Shor1997}.
Therefore, when a quantum computer becomes available for an adversary, the widely used algorithms for key distribution would no longer provide information security.
Besides, the impossibility of the existence of efficient classical algorithms for solving these problems has also not been proven, and it is only a conjecture.
There are currently no quantum computers capable of implementing Shor's quantum algorithm for fairly large numbers up to date, but such devices may appear in the near future.
Also, alternative approaches to solving the factorization problem are being developed, for example,
using variational quantum algorithms~\cite{Aspuru-Guzik2018}.
Thus, ensuring the confidentiality of data requires an early transition to cryptographic methods that would be resistant to attacks using a quantum computer~\cite{Mosca}.

One possible solution for key distribution in the `post-quantum era' is to use quantum key distribution (QKD), 
which was proposed by C.~Bennett and G.~Brassard in 1984~\cite{BB84}, 
and also independently by A.~Ekert in 1991~\cite{Ekert1991}.
Quantum key distribution is based on the idea of encoding information into the quantum states of individual quantum systems. 
For example, the most common QKD protocol BB84~\cite{BB84} uses photon states (for example, polarization) in two mutually unbiased bases~\cite{Holevo2018}.
The idea of using mutually unbiased bases was introduced by S.~Wiesner in the 1970s (his paper~\cite{Wiesner} was published later, in 1983) within the concept of ``quantum money''.
Fundamental limitations of quantum mechanics, such as the no-cloning theorem for a quantum state and the Heisenberg uncertainty principle, restrict the ability to read quantum information without changing it.

The first complete mathematical proof of the security of the BB84 protocol
was obtained in 1996 by D.~Mayers~\cite{Mayers0,Mayers}. 
Further proofs have followed, for example, Refs.~\cite{ShorPreskill, Koashi}.
The general mathematical theory of quantum key distribution, which is based on information theory and entropy characteristics, was developed by R.~Renner~\cite{Renner} 
and then improved in subsequent papers, 
among which we highlight Refs.~\cite{GisinRenner, Toma}.
One recent result in this area is the development of the ``entropy accumulation'' technique~\cite{EntroAccum}, 
which allows one to map the security proof for the case of so-called collective attacks (easier to analyze) to the case of general attacks (coherent attacks) that are difficult to analyze directly~\cite{Gisin2002,Scarani2009}.
For the case of BB84, it simplifies the security proof and also allows one to prove security against general attacks more generally, for example, in the case of detection-efficiency mismatch~\cite{Bochkov2019}.
For the operational meaning of the security parameter used in quantum key distribution, see Refs.~\cite{Portmann,TrushQElectronics}.

However, in practical implementations of QKD protocols, 
differences from ideal abstract protocols naturally arise~\cite{Scarani2009,L0216}.
These differences have a significant impact on the cryptographic security of implementation.
For example, the BB84 protocol assumes the encoding of information into states of single photons.
However, the generation of single photons ``on demand'' and at a high rate is a technically difficult task; 
therefore, instead of single photons, QKD uses weak coherent pulses~\cite{Huttner1995}.
This makes it possible to realize a photon-number splitting attack~\cite{Sanders2002}: 
The adversary can measure the number of photons in a pulse in a non-disturbing way, take one photon from all multiphoton states, and block all single-photon (i.e.~secure) states, or at least some of them.
This allows the adversary to extract the entire key or a significant part of it without introducing errors, 
which violates the basic property of the protocol about the relationship between the error level and the amount of intercepted information and, thus, makes the protocol implementation unreliable.

The vulnerability of quantum key distribution to such an attack can be eliminated by using the decoy state method, which was proposed in Refs.~\cite{Hwang2003,Lo2005,Wang2005,MaLo2005}.
The method is based on the fact that the sender does not use one fixed value for the intensity of coherent states, but each time randomly selects the intensity from a known finite set (also known to the adversary).
One of these intensities (the largest) is the signal one, and it is used to generate a key; the rest are decoy, and they are used to estimate the level of the adversary's intervention in the multiphoton pulses.

Being just a parameter in the probability distribution (namely, the Poisson distribution), the intensity is a non-observable quantity, so that only the photon number that is the realization of this random variable is observed.
Therefore, the adversary does not know what intensity was used in a given position and carries out his actions based only on the observed number of photons.

After finishing the transmission of quantum states, the sender announces the intensity of each position.
The legitimate parties then collect registration statistics for each intensity separately.
The states with lower intensities can be called decoy states in the sense, that for an adversary, 
who does not know the intensity parameter at the time of the attack, 
they are indistinguishable from the signal states, but after the intensities are announced, 
they become a sort of ``tagged'', which can be used to make separate statistics.

From a mathematical point of view, the statistics of detecting states with different intensities gives the legitimate parties additional equations for a better estimation of unknown parameters, 
such as the number of positions in the sifted key obtained from single-photon pulses 
(i.e., those that cannot be intercepted without introducing errors) and the fraction of errors in them.
In particular, blocking all single-photon states by the adversary would lead to the blocking almost all low-intensity decoy states.

Today, the decoy-state BB84 QKD protocol has been theoretically studied in detail~\cite{Hwang2003,Lo2005,Wang2005,MaLo2005,Lo2014,Lim2014,Ma2017,Trushechkin2017,Bochkov2019},
demonstrated in experiments~\cite{Lo20142,L0216}, including Russian domestic systems~\cite{Duplinskiy2018}, and it is considered to be a candidate for the international standard~\cite{Standard}.
Nevertheless, doubts about its security against all possible attacks, not just photon-number splitting attacks, 
have been expressed~\cite{Molotkov2019,Molotkov20192}.
For this reason, in this article, we not only describe the decoy state method but also provide the formal security proof of the decoy state method against all possible attacks.
This fact is usually not explained in the literature, since it is regarded obvious.
The proof would not refer to photon-number splitting attack.
However, as demonstrated, the photon-number splitting attack is optimal for the adversary, which explains the fact that countermeasures only against this attack are considered in the literature.

Separately, we compare the photon-number splitting attack with another common attack, namely, 
the beam splitting attack, and clearly demonstrate the lower efficiency of the latter.
At the same time, at realistic levels of losses, as it has been shown, these attacks give similar results.

The text is organized as follows.
In Sec.~\ref{sec:SecProtocol}, the BB84 protocol is described.
In Sec.~\ref{sec:SecPNS}, the problem in this protocol, which is related to the encoding information using coherent instead of single-photon states, is discussed, and the photon-number splitting attack is considered.
In order to formulate further results concerning the security of the protocol, 
in Sec.~\ref{sec:SecRate}, the concepts of achievable and maximally achievable secret key rates are introduced, 
and the well-known Devetak--Winter formula~\cite{DW} for the maximal secret key rate is presented.
Sec.~\ref{sec:SecReduc} is devoted to reducing the multiphoton case to the single-photon case if the transmitted states are statistical mixtures of Fock states (states with a certain number of photons).
The theorem in this section formally justifies the decoy state method and its security against all possible collective attacks.
Here, the optimality of the photon-number splitting attack is justified. 
Sec.~\ref{sec:SecEntroAccum} completes this analysis: 
An equivalent formulation of the protocol in terms of an entangled state allows using the entropy accumulation technique and justifies the security of the protocol to all possible (not just collective) attacks.
Sec.~\ref{sec:SecDecoy} is devoted to describing the decoy state method. 
In Sec.~\ref{sec:SecBS}, another attack, namely the beam splitting attack, is considered, which, unlike the photon-number splitting attack, has been implemented experimentally.
In Sec.~\ref{sec:SecProof}, we compare the photon-number splitting and beam splitting attacks and explain why the latter is less efficient. 
Finally, in Sec.~\ref{sec:SecPhaseCoding}, we respond to the doubts expressed in Ref.~\cite{KulikMol} about the applicability of the decoy state method in the case of the phase encoding rather than polarization encoding, 
and demonstrate the complete equivalence of these encodings.

It should be noted that the field of quantum information processing has already become the subject of several reviews in {\it Physics-Uspekhi}.
In particular, a number of problems in quantum computing was considered in Refs.~\cite{Kilin1999,Valiev2005}, 
and several aspects of quantum key distribution were considered in Ref.~\cite{Molotkov2006}.
The progress of recent decades on the development of industrial devices for quantum key distribution has taken research in this area to the next level, 
where a number of important practical aspects of the implementation of such systems have come to the fore.
In particular, one of the central problems of this area is to analyze the resistance to attacks of quantum key distribution protocols, 
which is the subject of this review.
The main result is a rigorous justification of the security of the decoy state method in quantum cryptography.

\section{BB84 protocol}\label{sec:SecProtocol}

In this section, we describe the BB84 protocol~\cite{BB84} under the assumption of a single photon source on the sender side.
Each QKD protocol can be divided into two main stages: 
Transmission of quantum states and post-processing of the measurement results.
Following Ref.~\cite{Kuzmin2002}, the communication participants that desire to obtain a shared key we call the sender and the receiver, while the eavesdropper side is the adversary.
Sender and receiver together are also referred to as legitimate parties.

In the first ``quantum'' stage of the BB84 protocol, four quantum states are used, which form two orthogonal bases $z{=}\{\ket0_z,\ket1_z\}$ and $x{=}\{\ket0_x,\ket1_x\}$ in the two-dimensional Hilbert space $\mathbb{C}^2$.
The quantum system that corresponds to this space is called the ``quantum bit'' or ``qubit''. 
The values 0 or 1 indicate which classic bit is encoded by the corresponding basis vector. 
Elements of the bases are expressed in terms of elements of another basis according to the relations
\begin{equation}\label{eq:bb84basis2}
	\ket0_x=\frac{\ket0_z+\ket1_z}{\sqrt2},
	\qquad
	\ket1_x=\frac{\ket0_z-\ket1_z}{\sqrt2}.
\end{equation}

If the information is encoded into photon polarization, then the vectors $\ket0_z$ and $\ket1_z$  can correspond, for example, to horizontal and vertical polarizations. 
In this case, $\ket0_x$ and $\ket1_x$ correspond to two diagonal polarizations that are rotated by 45$^\circ$ and 135$^\circ$ degrees, respectively, relative to the horizontal direction. 
We assume polarization coding for convenience of the presentation, but, in fact, no restriction is imposed on the method of information encoding: 
Formally, $\ket0_z$, $\ket1_z$, $\ket0_x$, and $\ket1_x$ are vectors in the Hilbert space and one can use any encoding which fulfills relation (\ref{eq:bb84basis2}).
In particular, we explain the equivalence of the polarization and phase encodings below in Sec.~\ref{sec:SecPhaseCoding}.

As is seen from (\ref{eq:bb84basis2}), when measuring a qubit in a basis different from the preparation basis, the result is a random value. 
In the case of coincidence of the preparation and measurement bases, 
the result perfectly correlates with the prepared state of the qubit 
(in the ideal case, i.e., in the absence of errors in the channel and measuring devices).

Let us now describe the BB84 protocol.
\begin{enumerate}
	\item 
	The sender randomly chooses a basis from the set $\{z,x\}$ and the value of the transmitted bit of information: $1$ or $0$. 
	Bits are selected with equal probabilities of 1/2.
	\item
	Then, the photons prepared in the corresponding states are transmitted through the quantum channel.
	\item
	The receiver randomly chooses a measurement basis, $z$ or $x$, for each qubit and measures the state of the qubit in the selected basis.
	If the preparation and measurement bases coincide, the received bit value coincides (ideally) with the sent one.
	If the bases do not coincide, the bits of the sender and the receiver do not correlate (that is, they may or may not coincide with equal probabilities) due to the fact that the bases are mutually unbiased (\ref{eq:bb84basis2}). 
    Usually, the communication channel contains large losses; 
    therefore, not all positions are registered by the receiver.
	\item
	The above steps are repeated many times, i.e., a large number of quantum states are transmitted. 
	As a result, legitimate parties receive two sequences of bits $k^{\rm raw}_A$ and $k^{\rm raw}_B$, which are called the {\it quantum keys}.
\end{enumerate}

Since a perfect copy of a quantum state cannot be created and the adversary does not know the basis in which the bit is encoded in a given position, the adversary needs to employ imperfect copying techniques that induce distortions.

In the original version of the protocol, the bases are chosen with equal probabilities.
Later, an improved variant of the protocol was suggested, in which one of the bases 
(for example, the $z$ basis) is chosen more frequently than the other one~\cite{AsymBB84}. 
This reduces the number of basis mismatches and, therefore, 
the portion of sifted positions, i.e., it increases the key rate. 
Let us denote the probabilities of choices of the bases as $p_z$ and $p_x=1-p_z$.
In the limit of an infinite number of pulses $N$, one can put $p_z\to1$, $p_x\to0$. 
For example, one can take $p_x=O(1/\sqrt N)$: 
This is enough for the statistical estimation of the parameters related to observations in the $x$ basis, since the statistical fluctuations have an order of $O(1/\sqrt{N})$. 
Then, the following modification of the protocol can be adopted: 
Only positions in which both parties used the $z$ basis are used for key generation.
The bits encoded in the $x$ basis do not participate in the formation of the secret key; 
they are only needed to estimate the level of eavesdropping.
Note that a version of the protocol in which the bases are chosen pseudo-randomly 
using a pre-distributed random sequence is considered in Ref~\cite{Trushechkin2018}.

At the second stage, the legitimate parties carry out the classical post-processing of raw keys using communication a over public authenticated channel~\cite{Gisin2002,Fung2018,Kiktenko2016}. 
It consists of the following steps:
\begin{enumerate}
	\item 
	{\it Announcements}. 
	The receiver announces the position numbers in which the signal was registered.
	The sender and receiver declare the bases used in all positions.
	When using the decoy state method, the sender also announces the type of each pulse (signal or decoy).
	The sender and the receiver can also announce bits in positions that do not participate in the formation of the secret key: 
	In positions in which the parties used the $x$ basis and in the decoy pulses.

	\item
	{\it Key sifting}.
	Positions in which decoy intensity were used, registration did not occur, or at least one of the legitimate parties used the $x$ basis are sifted out.
	The resulting keys, $k^{\rm sift}_A$ and $k^{\rm sift}_B$, are called {\it sifted keys}.
	Ideally, they should match, but as a result of natural noise in the channel or adversary actions, 
	they do not match.
	Moreover, the adversary may have partial information about them.
	
	\item
	{\it Error correction}.
	One of the sifted keys (for example, belonging to the sender) is considered to be a reference. 
	Differences between it and the sifted key of the other side are considered to be errors.
	One can use error correction codes or interactive error correction procedures to correct errors.
	Low-density parity-check (LDPC) codes are quite common.
	Often, this procedure ends with {\it verification}: 
	The identity of the sifted keys is checked using hash functions (see Ref~\cite{Fedorov2017}).
	As a result of this stage, the legitimate parties receive identical {\it verified keys} 
	$k^{\rm ver}_A=k^{\rm ver}_B$ with a high probability.
	An efficient method for error correction in the BB84 protocol error correction based on LDPC codes is described in Ref.~\cite{Kiktenko2017}; see also Ref.~\cite{Kronberg2017}.

	\item
	{\it Estimation of the level of eavesdropping} and making a decision about creating the key or renouncing it (aborting the protocol) based on the observed data.
	Quantum cryptography is based on the fact that information encoded in non-orthogonal quantum states cannot be read by a third party 
	(which does not know the basis in which the key bit in a given position was encoded) 
	without ``spoiling'' these states.
	Therefore, interception by the adversary would lead to an increase in the number of errors (i.e., mismatched positions in sifted keys) between the legitimate parties.
	In this version of the protocol, where only the bits encoded in the $z$ basis are involved in the formation of the key, only the fraction of errors in the $x$ basis is needed to assess the 
	level of eavesdropping.
	If the error rate exceeds a certain critical threshold, the protocol is aborted.
	Otherwise, the parties proceed to the last step.
	Thus, in quantum cryptography, it is impossible to eavesdrop without being detected.
	
	\item 
	{\it Privacy amplification.} 
	The sender randomly chooses a so-called hash function from a family of $2$-universal hash functions and sends it to the receiver via a public channel.
	Both then compute the hash value of their (identical) sifted keys.
	As a result, they get a common shorter key (\textit {final key}) $k^{\rm fin}_A=k^{\rm fin}_B$, but the information of the adversary about which is negligible.
	With an infinitely large length of the sifted key, it can be made arbitrarily small.
	The more information the adversary has about the sifted key (as a result of eavesdropping and as a result of disclosure by legitimate users of some of the information during error correction),
	the more compression of the key in the privacy amplification procedure is required, 
	i.e., the shorter the final key and the lower the key rate.
\end{enumerate}

As can be noted, the post-processing procedure requires the parties to communicate via the classical channel.
It is assumed that the adversary is freely able to listen to this channel but cannot change the messages transmitted over it or send adversary's one.
This is provided by message authentication codes.
In quantum cryptography, the information-theoretically secure Wegman-Carter codes~\cite{WegmanCarter1981} are often used (i.e., provably secure without assumptions about the adversary's computing power).
They require legitimate parties to have an initial short secret key.
It is sufficient to have an initially shared secret for the first key distribution session.
Further, in each session, part of the generated key is kept for use in a message authentication code in the next session and is not used for other purposes.
See Ref.~\cite{Kiktenko2019} for the latest developments in the reduction of the portion of the key 
that is consumed during authentication.

Here we assumed that the sender is sending single-photon states.
This corresponds to the fact that thier Hilbert space is $\mathbb{C}^2$.
At the end of the Introduction, there are references regarding the history of the security proof for the single-photon BB84 protocol.

\section{Information encoding in weak coherent pulses. Photon-number splitting attack}\label{sec:SecPNS}

Let us now take into account that, in practice, information is most often encoded not into true single-photon states but into weak coherent pulses.
This means that the Hilbert space of the sender is not $\mathbb C^2$, but $\mathscr F(\mathbb C^2)$ 
(the boson Fock space over $\mathbb C^2$).
The orthonormal basis in this space is formed by the vectors $\ket{\rm vac}$ and $\ket{j_0,j_1}_z$, where $\ket{\rm vac}$ is the vacuum state, and
\begin{equation}\label{Eqj0j1z}
	\ket{j_0,j_1}_z=\frac{(a^\dag_{z0})^{j_0}(a^\dag_{z1})^{j_1}}{\sqrt{j_0!j_1!}}\ket{\rm vac},
\end{equation} 
$j_0,j_1\geq0$ ($\ket{0,0}_z=\ket{\rm vac}$), $a^\dag_{z0}$, $a^\dag_{z1}$, $a_{z0}$, $a_{z1}$ are creation and annihilation operators of a photon in the states $\ket0_z,\ket1_z\in\mathbb{C}^2$, respectively. 
Another orthonormal basis is formed by vectors $\ket{\rm vac}$ and $\ket{j_0,j_1}_x$, $j_0,j_1\geq0$, where 
\begin{equation}\label{Eqj0j1x}
	\ket{j_0,j_1}_x=\frac{(a^\dag_{x0})^{j_0}(a^\dag_{x1})^{j_1}}
	{\sqrt{j_0!j_1!}}\ket{\rm vac},
\end{equation}
\begin{equation}\label{Eqax}
	a^\dag_{x0}=\frac{a^\dag_{z0}+a^\dag_{z1}}{\sqrt2}, \qquad 
	a^\dag_{x1}=\frac{a^\dag_{z0}-a^\dag_{z1}}{\sqrt2}
\end{equation}
are creation and annihilation operators of photons in states $\ket0_x,\ket1_x\in\mathbb C^2$, respectively. 
In particular, vectors $\ket 0_b, \ket1_b\in\mathbb C^2$ can be identified with vectors $\ket{1,0}_b,\ket{0,1}_b\in\mathscr F(\mathbb C^2)$, respectively, $b\in\{z,x\}$.

Then, the transmitted coherent state has the form
\begin{equation}\label{EqCohState}
	\ket{\alpha;u}_b=e^{-\mu/2}\ket{\rm vac}+
	e^{-\mu/2}\sum_{j=1}^\infty 
	\frac{\alpha^j}{\sqrt{j!}}\ket{\psi^b_{ju}},
\end{equation}
where $u\in\{0,1\}$, $b\in\{z,x\}$,
\begin{equation*}
	\ket{\psi^b_{j0}}=\ket{j,0}_b,\quad \ket{\psi^b_{j1}}=\ket{0,j}_b,
\end{equation*}
$b\in\{z,x\}$. 
Here $\alpha\in\mathbb C$ is the parameter of the coherent state, 
$\mu=|\alpha|^2$ is the pulse intensity, $\alpha=\sqrt{\mu}e^{i\theta}$. 
The protocol demands a random change in the phase $\theta$ of the coherent state from pulse to pulse. 
This is provided either by a laser operating mode (passive randomization)
or by introducing an additional element, 
which is connected to a random number generator and randomizes the phase, 
into the optical system of the sender (active randomization)
(see also Remark~\ref{RemPhRnd} in Sec.~\ref{sec:SecReduc}).
Then, the adversary and the receiver do not ``observe'' a pure coherent state but a mixed one with the density operator
\begin{equation}\label{EqRhoPhRnd}
\begin{split}
	\rho^b_{\mu u}&=\frac1{2\pi}
	\int_0^{2\pi}d\theta\,
	\ket{\sqrt{\mu}e^{i\theta};u}_b\bra{\sqrt{\mu}e^{i\theta};u} \\
	&=
	e^{-\mu}\ket{\rm vac}\bra{\rm vac}+
	e^{-\mu}\sum_{j=1}^\infty\frac{\mu^j}{j!}\ket{\psi^b_{ju}}\bra{\psi^b_{ju}}.
\end{split}
\end{equation}
Here, for simplicity, we assume that the basis $b$ and the bit $u$ are fixed; otherwise, it is necessary to average over them.
This mixed state can be interpreted as follows: The vacuum state is sent with the probability $e^{-\mu}$, whereas the state with the $j$ photon in the corresponding polarization is sent with the probability $e^{-\mu}\mu^j/j!$.
Thus, the number of photons in the pulse is distributed according to the Poisson distribution with the parameter (average number of photons) $\mu$.
Usually, $\mu<1$ is chosen; therefore, such impulses are referred to as weak coherent.

The presence of more than one photon in some pulses allows the adversary to carry out a so-called {\it photon-number splitting attack}~\cite{Sanders2002}.
Sometimes this is understood as an attack of a special type, but we mean a whole class of attacks, which we describe now.
Each attack of this class begins with the adversary measuring the number of photons, 
i.e., applying an observable which corresponds to the probability projection-valued measure $\{P_j\}_{j=0}^\infty$, where 
\begin{equation}\label{EqQND}
\begin{split}
	P_j=&\sum_{j_0+j_1=j}\ket{j_0,j_1}_z\bra{j_0,j_1}= \\
	&\sum_{j_0+j_1=j}\ket{j_0,j_1}_x\bra{j_0,j_1}.
\end{split}
\end{equation}
For an outcome $j$, state (\ref{EqRhoPhRnd}) transforms into 
$P_j\rho_{\mu u}^b P_j/\Tr(\rho_{\mu u}^b P_j)$.
This measurement is also called nondemolition, because the photons are not destroyed, 
and their polarization is not changed.
In practice, such a nondemolition measurement has not yet been implemented, 
but it is possible in theory, since the corresponding probability projector-valued measure exists.
If the adversary observes that there are two or more photons in the transmission ($j\geq2$), 
then they take one photon into thier quantum memory and send the rest to the receiver.
After the bases are announced via the public channel, 
the adversary would find out which basis in a given position the key bit is encoded.
They measure with this basis and recognize this bit. 
In this case, the state of those photons that were sent to the receiver does not change.
This violates the basic principle of quantum cryptography that the observation leads to changes in the state and occurrence of (or increasing) errors.
If the pulse contains one photon, the adversary can still find out the information encoded in it only at the cost of changing the state and, accordingly, introducing errors.
Therefore, the adversary can block all or part of one-photon states, i.e., 
stop their transmission to the receiver, simulating natural losses in the communication channel.
They can attack unblocked single-photon states in the usual way (at the cost of introducing errors).

The blocking of some pulses by the adversary leads to an increase in the level of losses, 
which can be detected by the legitimate parties.
Therefore, it is assumed that the adversary can replace the communication line and even detectors with ideal ones and then block as many single-photon pulses as possible to reproduce the natural level of loss.
The higher the level of natural losses (which depends primarily on the length of the communication line, see Sec.~\ref{sec:SecDecoy}), the more single-photon states the adversary can block.
If the natural losses in the channel are so large that the adversary can block all single-photon states,
then the adversary would have full information about the key without introducing noise, since all the signals that reach the receiver would be multiphoton.
Thus, the quantum key distribution protocol is completely compromised.

\section{Secret key rate}\label{sec:SecRate}

The {\it final (secret) key rate} is understood as the limit of the ratio of the length of the final key to the length of the sifted key, as the number of pulses tends to infinity.
The {\it achievable} secret key rate means the secret key rate at which it is possible to ensure the requirement of an infinitesimal amount of information of the adversary about the final key (in the mentioned limit).
As mentioned above, this requires a certain level of compression in the privacy amplification procedure.
Of course, if a certain rate is achievable, then lower rates are also achievable.
The {\it maximally achievable} (or simply {\it maximal}) secret key rate is called the exact upper bound of the set of achievable rates.
For formal definitions, see Ref.~\cite{DW}.

Note that sometimes the key rate is defined as the limit of the length of the final key to the number of pulses.
Since the length of the sifted key is proportional to the number of pulses, 
these definitions coincide up to a constant factor.

The formula for the maximum secret key rate can be written using the Devetak--Winter theorem~\cite{DW} 
(see also section~\ref{sec:SecEntroAccum}).
Its formulation involves a three-particle state of the sender, the receiver, and the adversary, 
which emerges after the sifting procedure.
Our immediate goal is to write the formula for this state.

Let us introduce a register (formally, a quantum system) $\overline{A}$ of dimension 2, 
which stores the bit value of the sender.
We denote a subsystem that contains the transmitted quantum state~(\ref{EqRhoPhRnd}) as $A$.
Then, the joint state $\overline A$ and $A$ when using the $z$ basis has the form
\begin{equation}\label{EqRhoAA}
	\rho_{\overline AA}^z=\frac12\sum_{u=0}^1
	\ket u_{\overline A}\bra u\otimes\rho^z_{\mu u},
\end{equation}
where $\rho^z_{\mu u}$ is given by Eq.~(\ref{EqRhoPhRnd}).

Due to the action of natural noise, natural losses in the channel, and/or actions by the adversary, the state $\rho^z_{\mu u}$ undergoes the action of some quantum channel,
i.e., linear completely positive trace-preserving map $\Upsilon_0\colon\mathfrak T(\mathscr H_A)\to\mathfrak T(\mathscr H_B\otimes\mathscr H_E)$,
where $\mathscr H_A=\mathscr H_B=\mathscr F(\mathbb C^2)$ are Hilbert spaces of the sender and the receiver.
$\mathscr H_E$ is some (unknown) Hilbert space of the adversary, and
$\mathfrak T$ is the space of nuclear operators in the corresponding Hilbert space.

The quantum efficiencies of the detectors (which are different in general) and the probabilities of dark counts (also, in general, different) can also be included in the channel $\Upsilon_0$:
In the first case, the state with photons turns, with some probability, into the state without photons, and the other way around in the second case.
It should be noted that enabling these effects makes $\Upsilon_0$ basis-dependent.
The transformation of the state in the communication line does not depend on the basis, since the adversary does not know the basis at the time of the attack, and the natural noise and losses do not depend on the basis.
But the further transformation of the state associated with the measurement is already dependent on the measurement basis.
We are now interested in positions in which both sides use the $z$ basis, since only such positions participate in the key formation.
The measurement in the $x$ basis would correspond to another transformation, but we will not need it.

After this, an ``ideal'' measurement (we have already taken into account imperfect detector efficiencies and dark counts in $\Upsilon_0$) corresponds to the probability projector-valued measure
$\{\bar P_{\varnothing},\bar P_0,\bar P_1,\bar P_{01}\}$, where $\bar P_\varnothing=\ket{\rm vac}\bra{\rm vac}$ and $\bar P_\varnothing=\ket{\rm vac}\bra{\rm vac}$ corresponds to the no click event in both detectors,
$$
	\bar P_0=\sum_{j=1}^\infty\ket{j,0}_z\bra{j,0}\quad \text{ and }\quad \bar P_1=\sum_{j=1}^\infty\ket{0,j}_z\bra{0,j}
$$ 
correspond to the click of exactly one detector, 
and $\bar P_{01}=\sum_{j,k=1}^\infty\ket{j,k}_z\bra{j,k}$ corresponds to the click of both detectors. 

So, the vacuum component of the receiver state is not registered.
Therefore, the corresponding positions are sifted out and do not participate in the key formation.
Recall that our immediate goal is to write a formula for the three-particle state of the sender, 
the receiver, and the adversary, which emerges after the sifting procedure.
We denote
\begin{equation*}
\begin{split}
	\Upsilon(\rho_A)=[(I_B-\ket{\rm vac}\bra{\rm vac})\otimes I_E]\times \\
	\times\Upsilon_0(\rho_A)\,
	[(I_B-\ket{\rm vac}\bra{\rm vac})\otimes I_E]
\end{split}
\end{equation*}
for any $\rho_A\in\mathfrak T(\mathscr H_A)$. 
Here $I_B$ and $I_E$ are the identity operations in the spaces $\mathscr H_B$ and $\mathscr H_E$, respectively.
Obviously, the map $\Upsilon$ does not preserve the trace; It only does not increase it: $\Tr\Upsilon(\rho_A)\leq\Tr\rho_A$. 
The reason for the non-conservation of the trace is that $\Upsilon$ includes sifting, which reduces the trace. 

Then, the desired formula for the three-particle state of the sender, receiver, 
and adversary after sifting has the form
\begin{equation}\label{EqRhoABE}
\begin{split}
	\rho_{\overline ABE}&=(Q^{{\rm s}z})^{-1}({\rm Id}_{\overline A}\otimes\Upsilon)(\rho_{\overline AA}^z)=\\ 
	&=\frac1{2Q^{{\rm s}z}}\sum_{u=0}^1
	\ket u_{\overline A}\bra u\otimes\Upsilon(\rho^z_{\mu u}),
\end{split}
\end{equation}
where 
\begin{equation*}
	Q^{{\rm s}z}=\Tr\,({\rm Id}_{\overline A}\otimes\Upsilon)(\rho_{\overline AA}^z)
\end{equation*}
is the probability of the registration of a signal pulse (as indicated by the superscript ${\rm s}$, later we would need the probabilities of registration of decoy pulses),
and ${\rm Id}_{\overline A}$ is the identity quantum map of density operators in the space of register $\overline A$, i.e. $\mathbb C^2$.

According to the Devetak--Winter theorem, the maximal secret key rate is as follows:
\begin{equation}\label{EqDW}
	R=
	H(\overline A|E)-
	H(\overline A|B),
\end{equation}
where 
\begin{equation}\label{EqQCondEntro}
	\begin{split}
	H(\overline A|B)&=H(\rho_{\overline AB})-H(\rho_B),
	\\
	H(\overline A|E)&=H(\rho_{\overline AE})-H(\rho_E),
\end{split}
\end{equation}
are quantum conditional entropies.

Here we have used the following convention regarding the denotations of the states of subsystems of a composite system: If the state of the composite system is $\rho_{\overline ABE}$,
then $\rho_{\overline AB}=\Tr_{E}\rho_{\overline ABE}$, $\rho_B=\Tr_{\overline AE}\rho_{\overline ABE}$, $\rho_E=\Tr_{\overline AB}\rho_{\overline ABE}$, etc.
In turn, $H(\rho)=-\Tr(\rho\log\rho)$ is the von Neumann entropy, $\log\equiv\log_2$. 
Since the state $\rho_{\overline{A}BE}$ is classical-quantum (it is classical on $\overline{A}$ and quantum on $BE$), both terms on the right-hand side of Eq.~(\ref{EqDW}) are non-negative.

The first term on the right-hand side of Eq.~(\ref{EqDW}) characterizes the {\it ignorance} or 
the {\it lack of information} of the adversary about the sender's key bit.
The second term (without the minus sign before it) characterizes the 
ignorance (the lack of information) of the receiver about the same bit.
Denote $m$ as the minimal amount of information that the sender has to disclose about thier sifted key, so that the receiver would be able to correct all errors and get a key that matches the sender's key.
If error-correcting codes are used for this purpose, then $m$ is the length of the syndrome.
However, iterative error correction procedures can also be used.
The second term on the right-hand side of (\ref{EqDW}) is the ratio $m/n$ as $n\to\infty$. 
We assume that an optimal error correction code (in the Shannon sense) is used.
Otherwise, the factor $f>1$ has to be added before the second term.
So, for example, the value of the factor $f=1.22 $ is achievable for modern error-correcting codes~\cite{Kiktenko2017}.

According to the Fano inequality, the quantity $H(\overline A|B)$ can be upper bounded by 
$h(E^{{\rm s}z})$, where
\begin{equation*}
	h(x)=-x\log x-(1-x)\log(1-x)
\end{equation*}
is the binary entropy, and $E^{{\rm s}z}$ is the fraction of errors in the sifted keys. 
The superscripts ${\rm s}$ and $z$ indicate that we are talking about signal states and positions, 
in which both sides used the basis $z$.
Recall that only these positions are involved in the secret key formation.
This value becomes known to the legitimate parties after error correction and key verification: After verification, the keys are likely to coincide; therefore, by the number of positions in which corrections have taken place, the legitimate parties know the error rates in the sifted keys.
Then, to obtain a formula for the achievable secret key rate, 
the legitimate parties have to evaluate the first term on the right-hand side of Eq.~(\ref {EqDW}), 
i.e., the adversary's ignorance.
Estimation of the first term is crucial in proving the security of any QKD protocol.

\section{Reducing the multiphoton case to a single-photon one}\label{sec:SecReduc}

Let us reduce the problem of upper bounding $H(\overline{A}|E)$ for the state $\rho_{\overline{A}BE}$,
which includes multiphoton pulses, 
to an estimation of the quantum conditional entropy for a state that includes only single-photon pulses.
We write
\begin{align*}
	&\widetilde\rho^{(0)}_{\overline AA}=\frac12\sum_{u=0}^1
	\ket u_{\overline A}\bra u\otimes P_0\rho^z_{\mu u}P_0= \\
	&=\frac{e^{-\mu}}2
	\sum_{u=0}^1\ket u_{\overline A}\bra u \otimes\ket{\rm vac}\bra{\rm vac},\\
	&\widetilde\rho^{(1)}_{\overline AA}=\frac12\sum_{u=0}^1
	\ket u_{\overline A}\bra u\otimes P_1\rho^z_{\mu u}P_1= \\
	&\frac{\mu e^{-\mu}}2
	\sum_{u=0}^1\ket u_{\overline A}\bra u \otimes\ket u_z\bra u,\\
	&\rho_{\overline{A}BE}^{(j)}=
	(Q^{{\rm s}z}_j)^{-1}
	\,({\rm Id}_{\overline A}\otimes\Upsilon)
	(\widetilde\rho^{(j)}_{\overline AA}), \\
	&
	Q^{{\rm s}z}_j=\Tr\,({\rm Id}_{\overline A}\otimes\Upsilon)
	(\widetilde\rho^{(j)}_{\overline AA}),\quad j=0,1,
\end{align*}
where  $P_j$ is the projection onto the $j$-photon subspace in $\mathscr F(\mathbb C^2)$, see Eq.~(\ref{EqQND}), 
and $Q^{{\rm s}z}_j$ is the joint probability that the sender's signal pulse contains exactly $j$ photons and is registered by the receiver. 
Then, the ratio of the $j$-photon pulses to the pulses registered by the receiver and
participating in the formation of the sifted key is $Q^{{\rm s}z}_j/Q^{{\rm s}z}$.

\begin{theorem}\label{ThDecoy}
For any linear completely positive trace-nonincreasing map  $\Upsilon$, the following inequality holds:
\begin{equation}\label{EqReduc}
    H(\overline A|E)\geq 
    \frac{Q^{{\rm s}z}_0}{Q^{{\rm s}z}}
    +
    \frac{Q^{{\rm s}z}_1}{Q^{{\rm s}z}}
    H(\overline A|E)_{(1)},
\end{equation}
where $H(\overline A|E)$ is calculated for state $\rho_{\overline{A}BE}$ and $H(\overline A|E)_{(1)}$ is calculated for state $\rho_{\overline{A}BE}^{(1)}$.
\end{theorem}

\begin{proof}
Quantum conditional entropy can be expressed in terms of the quantum relative entropy:
\begin{equation*}
	H(\overline A|E)=-D(\rho_{\overline AE}\|I_{\overline A}\otimes\rho_E).
\end{equation*}
Recall that the quantum relative entropy is defined as 
$$
	D(\rho\|\sigma)=\Tr(\rho\log\rho)-\Tr(\rho\log\sigma),
$$
and it is jointly convex with respect to its arguments~\cite{Lieb,Holevo}
\begin{equation*}
\begin{split}
	D\Big(p\rho_1+(1-p)\rho_2\,\Big\|\,p\sigma_1+(1-p)\sigma_2\Big) 
	\leq \\
	\leq
	pD(\rho_1\|\sigma_1)+(1-p)D(\rho_2\|\sigma_2)
\end{split}
\end{equation*}
for any states $\rho_{1,2}$, $\sigma_{1,2}$ and $0\leq p\leq1$. 

Then, the statement of Theorem~\ref{ThDecoy} is a simple consequence of the fact that the state $\rho^z_{\mu u}$ is a mixture of states with a certain number of photons (\ref{EqRhoPhRnd}),
linearity of $\Upsilon$, and joint convexity of the quantum relative entropy. 
Indeed,
\begin{equation*}
\rho_{\mu u}^z=
	P_0\rho^z_{\mu u}P_0+
	P_1\rho^z_{\mu u}P_1+P_{\geq2}\rho^z_{\mu u}P_{\geq2},
\end{equation*}
where $P_{\geq2}=\sum_{j=2}^\infty P_j$. 
Therefore, in view of the linearity of $\Upsilon$,
\begin{equation*}
\begin{split}
	\rho_{\overline ABE}&=
	(Q^{{\rm s}z})^{-1}
	({\rm Id}_{\overline A}\otimes\Upsilon)(\rho_{\overline AA})=
	\\
	&=
	\frac{Q^{{\rm s}z}_0}{Q^{{\rm s}z}}\,\rho_{\overline ABE}^{(0)}+
	\frac{Q^{{\rm s}z}_1}{Q^{{\rm s}z}}\,\rho_{\overline ABE}^{(1)}+
	\frac{Q^{{\rm s}z}_{\geq2}}{Q^{{\rm s}z}}\,\rho_{\overline ABE}^{(\geq2)},
\end{split}
\end{equation*}
where
\begin{align*}
	\rho_{\overline{A}BE}^{(\geq2)}&=
	(Q^{{\rm s}z}_{\geq2})^{-1}
	\,({\rm Id}_{\overline A}\otimes\Upsilon)
	(\widetilde\rho^{(\geq2)}_{\overline AA}), \\
	Q^{{\rm s}z}_{\geq2}&=\Tr\,({\rm Id}_{\overline A}\otimes\Upsilon)
	(\widetilde\rho^{(\geq2)}_{\overline AA}),\\
	\widetilde\rho^{(\geq2)}_{\overline AA}&=\frac12\sum_{u=0}^1
	\ket u_{\overline A}\bra u\otimes P_{\geq2}\rho^z_{\mu u}P_{\geq2}.
\end{align*}
Then, due to the joint convexity of the quantum relative entropy,
\begin{equation*}
\begin{split}
	H(\overline A|E)&\geq
	\frac{Q^{{\rm s}z}_0}{Q^{{\rm s}z}}H(\overline A|E)_{(0)}+ \\
	&+
	\frac{Q^{{\rm s}z}_1}{Q^{{\rm s}z}}H(\overline A|E)_{(1)}+
	\frac{Q^{{\rm s}z}_{\geq2}}{Q^{{\rm s}z}}H(\overline A|E)_{(\geq2)} \\
	&\geq 
	\frac{Q^{{\rm s}z}_0}{Q^{{\rm s}z}}+
    \frac{Q^{{\rm s}z}_1}{Q^{{\rm s}z}}H(\overline A|E)_{(1)},
\end{split}
\end{equation*}
Q.E.D. 
Here conditional entropies $H(\overline A|E)_{(0)}$ and $H(\overline A|E)_{(\geq2)}$ are calculated for the states $\rho_{\overline{A}BE}^{(0)}$ and $\rho_{\overline{A}BE}^{(\geq2)}$, respectively. 
They are nonnegative, because these states are classical-quantum.
We have also used the fact here that $H(\overline A|E)_{(0)}=H(\overline A)=1$, which follows from the form of the states $\widetilde\rho^{(0)}_{\overline AA}$ and $\widetilde\rho^{(0)}_{\overline ABE}$. 
Informally, the vacuum contains no information about the sent bit.
\end{proof}

Theorem~\ref{ThDecoy} proves that the decoy state method is universal, i.e., 
secure against arbitrary attacks (for arbitrary maps $\Upsilon$), 
and not only against the photon-number splitting attack.
More precisely, Theorem~\ref{ThDecoy} claims the universality of the lower bound (\ref{EqReduc}), 
whereas the decoy state method provides lower  bounds for $Q^{{\rm s}z}_0$ and $Q^{{\rm s}z}_1$,
as well as bounds for quantities participating in the estimation of $H(\overline A|E)_{(1)}$.

Estimate (\ref{EqReduc}) 
becomes exact in the case of the photon-number splitting attack, since this attack is optimal.
Indeed, this attack provides knowledge of all the bits that are encoded in multiphoton pulses to the adversary,
which is the basic assumption of this estimate.

Is it always possible to perform a photon-splitting attack?
The removal of a photon from a multiphoton signal by the adversary reduces the registration probability of this signal, that is, increases the level of losses.
As mentioned above, the adversary can compensate for this by reducing the level of natural losses in the communication channel or even in the optical scheme of the receiving side.
However, if the natural losses are so small that even their reduction to zero does not compensate for the losses introduced due to the removal of photons from multiphoton signals,
then the adversary cannot perform a full photon-number splitting attack.
Therefore, estimate~(\ref{EqReduc}) ceases to be precise (but, of course, remains true):
$H(\overline A|E)_{(\geq2)}>0$. 
Nevertheless, at realistic losses, the adversary can perform a photon-number splitting attack, 
and estimate~(\ref{EqReduc}) is precise. 

Let us now pay attention to the first term of the right-hand side of Eq.~(\ref{EqReduc}). 
It is related to the registration of a vacuum pulse.
The registration of a vacuum pulse can occur due to the actions of the adversary, 
who can send their signal to the receiver.
It can also happen due to a dark count in the receiver's detectors.
However, if key distribution is possible, the fraction of dark counts in the total number of counts (i.e., the ratio  $Q^{{\rm s}z}_0/Q^{{\rm s}z}$) is small.
Key distribution is only possible with a relatively low error rate.
The well-known theoretical maximal tolerable error rate for the BB84 protocol is 11\% (under the assumption that the error rate is the same for both bases, see Remark~\ref{RemQber} below)~\cite{Mayers,ShorPreskill,Gisin2002}. 
In practice, the maximal tolerable error rate is even lower due to the presence of multiphoton pulses and due to statistical fluctuations.
Since a dark count is erroneous with a probability of 1/2, this means that the fraction of dark counts in the total number of counts on the receiver side does not exceed 6\%.
For this reason, instead of an accurate estimate (\ref{EqReduc}), a rougher estimate is often used:
\begin{equation}\label{EqReduc1}
	H(\overline A|E)\geq 
	\frac{Q^{{\rm s}z}_1}{Q^{{\rm s}z}}
	H(\overline A|E)_{(1)}.
\end{equation}
We emphasize that the fraction of dark counts affects the closeness of estimate~(\ref{EqReduc1}) to the exact estimate~(\ref{EqReduc}), but not on the validity of Eq.~(\ref{EqReduc1}), 
which follows directly from Eq.~(\ref{EqReduc}) and the positivity of values $Q^{{\rm s}z}_0$ and $Q^{{\rm s}z}$.
This means that Eq.~(\ref{EqReduc1}) allows secret key generation at any dark count rate,
but at a slightly lower speed than is allowed by formula~(\ref{EqReduc}).

In what follows, in the decoy state method, dark counts are strictly taken into account in the value $Y_0$, see below Eqs.~(\ref{EqQiv}), (\ref{EqQv}), and (\ref{EqY0L}). 
At realistic communication channel lengths and without eavesdropping, 
i.e., during the `normal' operation of the system,
the dark counts (rather than, for example, errors due to inaccurate alignment of the optical scheme) 
make the main contribution to the error rate.

\begin{remark}\label{RemPhRnd}
Here the requirement of phase randomization (\ref{EqRhoPhRnd}) is crucial.
As mentioned above, it can be achieved in two ways: At the hardware level, when a laser emits a pulse with a random phase,
independent of the phases of previous impulses, or using active randomization, i.e., with the use of a device connected to a random number generator for the random choice of the phases of outgoing pulses.
In the latter case, randomization according to a uniform or another continuous distribution is not possible: 
Only a discrete approximation is possible; therefore, it is necessary to take into account the corresponding correction~\cite{DiscrPhRnd}.
\end{remark}

The Devetak--Winter formula for the maximal secret key rate was initially obtained for the case of a so-called collective attack.
Keep in mind that attacks on quantum cryptography protocols are divided (by increasing degree of generality) into individual, collective, and coherent~\cite{Gisin2002}.
In a collective attack, it is assumed that the adversary attacks all the transmitted states in the same way, so that, after sending $n$ pulses, the state $\rho_{\overline ABE}^{\otimes n}$ is formed [see Eq.~(\ref{EqRhoABE})].
Then, the adversary performs a collective (i.e., of the general form) measurement over thier own parts.

In the most general case, the adversary performs an arbitrary quantum transformation over $n$ pulses, which is not necessarily a tensor power of the transformation over one pulse.
These attacks are called coherent attacks.
In the asymptotic case $n\to\infty $, which we restrict ourselves to in this paper, the consideration of coherent attacks is reduced to the consideration of collective attacks using the de Finetti quantum representation~\cite{Renner}.
Therefore, formally, the use of the Devetak--Winter formula is not a limitation.

However, the corrections associated with the finiteness of $n$, when using the de Finetti representation, turn out to be large enough for realistic $n$; therefore, in practice, the de Finetti representation is usually not used.
Let us explain why the consideration of the photon-number splitting attack does not limit the generality in this case either.

For $n$ pulses, instead of state (\ref{EqRhoAA}), we now have the state 
\begin{equation}
	\rho_{\overline{\mathbf A}\mathbf A}^z=
	\frac1{2^n}\sum_{\mathbf u\in\{0,1\}^n}
	\ket{\mathbf u}_{\overline A}\bra{\mathbf u}
	\otimes
	\rho^z_{\mu\mathbf u},
\end{equation}
where $\mathbf u=u_1\ldots u_n$ is a binary string, $\ket{\mathbf u}=\ket{u_1}\otimes\ldots\otimes\ket{ u_n}$,
\begin{equation*}
	\rho^z_{\mu \mathbf u}=\rho^z_{\mu u_1}\otimes\ldots\otimes\rho^z_{\mu u_n},
\end{equation*}
and bold marking of the registers $(\overline{\mathbf A}, \mathbf A)$ means that registers are ``vectors''. 

Since
\begin{equation*}
	\rho^z_{\mu u}=\sum_{j=0}^\infty
	P_j\rho^z_{\mu u}P_j,
\end{equation*}
it can be stated that
\begin{equation}\label{EqManySums}
	\rho^z_{\mu\mathbf u}=
	\sum_{j_1=0}^\infty\ldots\sum_{j_n=0}^\infty
	(P_{j_1}\rho^z_{\mu u_1}P_{j_1})
	\otimes\ldots\otimes
	(P_{j_n}\rho^z_{\mu u_n}P_{j_n}).
\end{equation}

Each term on the right-hand side of Eq.~(\ref{EqManySums}) represents a state (up to normalization) with a certain number of photons in each position.
Since transformations of quantum states are linear (or affine) maps,
any state transformation (\ref{EqManySums}) is a convex sum of the results of state transformations with a certain number of photons in each pulse.
Therefore, the consideration of only such states does not limit the generality.
Measurements of the number of photons in each pulse in a photon-number splitting attack do not change state~(\ref{EqManySums}) 
(since it already has a diagonal form with respect to the operators $P_{j_1}\otimes\ldots\otimes P_{j_n}$), 
but simply allow the adversary to find out the number of photons in each pulse.
In this sense, the photon-number splitting attack is optimal: 
If each pulse has a certain number of photons and one can find out this number without introducing noise,
then it is optimal to do this, that is, to measure the number of photons in each pulse.

A rigorous estimate of the entropy characterizing the ignorance of the adversary for a state of 
form~(\ref{EqManySums}) for a finite $n$ was performed in Ref.~\cite{ConciseDecoy}.
The purpose of this consideration was to justify the fact that the consideration of states with a certain number of photons in each pulse does not limit the generality.

\section{Formulation in terms of entangled states}\label{sec:SecEntroAccum}

For the generalization of the Devetak--Winter results for the case of coherent attacks, one can use the entropy accumulation technique~\cite{EntroAccum}.
It consists of estimating the smoothed Renyi min-entropy for a state in $(\mathscr H_A\otimes\mathscr H_B\otimes\mathscr H_E)^{\otimes n}$
(which is necessary to estimate the ignorance of the adversary during a coherent attack) through the von Neumann entropy, which is the first term on the right-hand side of the Devetak--Winter formula~(\ref{EqDW}).
Therefore, if it is possible to apply this technique, 
the Devetak--Winter formula gives the maximally achievable  secret key rate  in the case of coherent 
(i.e., general) attacks.
Also, the entropy accumulation technique makes it possible to take into account the effects associated with the finite size of the statistical sample (finite number $n$),
which is outside the scope of this work: 
Recall that we are working in the limit of infinitely large $n$.

In Ref.~\cite{EntroAccum}, this technique was used to prove the Devetak--Winter formula for the BB84 protocol with a single-photon source in the case of general attacks.
However, it uses an equivalent entangled state representation of the BB84 protocol.
Therefore, the purpose of this section is to formulate the BB84 protocol with phase-randomized coherent states in terms of an entangled state.

Recall that $\mathscr H_A$ denotes the Hilbert space, corresponding to a quantum information carrier: $\mathbb C^2$ in the single-photon case and $\mathscr F(\mathbb C^2)$ in the general case.
$\overline A$ denotes a binary register which stores the sender's bit value. 
Let us introduce yet another Hilbert space $\mathscr H_{\widetilde A}$, which is related to the sender.
In the BB84 protocol, the sender transmits the states $\rho_{\mu u}^b\in\mathfrak T(\mathscr H_A)$, $b\in\{z,x\}$, $u\in\{0,1\}$, of the from of Eq.~(\ref{EqRhoPhRnd}) with the probabilities $p_b/2$. 
To define the protocol in terms of entangled states, let us assume
that an entangled state $\rho\in\mathfrak T(\mathscr H_{\widetilde A}\otimes\mathscr H_A)$ of the composite system $\mathscr H_{\widetilde A}\otimes\mathscr H_A$ is generated on the sender side.
The sender measures subsystem $\widetilde A$, whereas subsystem $A$, as previously, is sent to the receiver via the communication channel.
Let the sender's observable of subsystem $\widetilde A$ be given by the probability projector-valued measure $\Pi=\{\Pi_{ub}\}_{u\in\{0,1\},b\in\{z,x\}}$ in the space $\mathscr H_{\widetilde A}$. 
We require the fulfilment of the following conditions
\begin{equation}
	\Tr_{\widetilde A}(\Pi_{ub}\rho\Pi_{ub})=\frac{p_b}2\rho_{\mu u}^b,
\end{equation}
for all $u$ and $b$, where $\Tr_{\widetilde A}$ is the partial trace in the space $\mathscr H_{\widetilde A}$. 
Then, the sender can prepare the states $\rho_{\mu u}^b$ with the corresponding probabilities by measuring the observable $\Pi$ of the subsystem $\widetilde A$, 
so that we obtain a mathematically equivalent scheme of the protocol. 
For its security, the state $\rho$ must be entangled; 
the title of this scheme of the protocol comes from this fact.
If, in reality, 
an entangled state is not generated and the sender does not measure one of the subsystems of an entangled state,
then the system $\mathscr H_{\widetilde A}$ and the measurement in it are  ``fictitious'', so they are used only for a mathematically equivalent formulation.

In the single-photon case (when $\mathscr H_A=\mathbb C^2$ and the states $\rho_{\mu u}^b$ contain the single-photon component only), one can take $\mathscr H_{\widetilde A}=\mathbb C^2$, $\rho=\ket\Phi\bra\Phi$,
where 
\begin{equation}\label{EqEPR}
\begin{split}
	&\ket{\Phi}=\frac1{\sqrt2}(\ket{1,0}_z\otimes\ket{1,0}_z+\ket{0,1}_z\otimes\ket{0,1}_z)
	= \\
	&=\frac1{\sqrt2}(\ket{1,0}_x\otimes\ket{1,0}_x+\ket{0,1}_x\otimes\ket{0,1}_x)
	\in\mathbb C^2\otimes\mathbb C^2,
\end{split}
\end{equation}
and $\Pi_{ub}=p_b\ket u_b\bra u$. 
The security of this protocol, which is based on the entangled state $\ket\Phi$, is proved in Ref.~\cite{EntroAccum} using the entropy accumulation technique.
In this case, in the result of the measurement of the observable $\Pi$, $\widetilde A$ 
stores the sender's bit value. 
So, register $\widetilde A$ can be identified with $\overline A$. 
However, generally speaking, they are not identical.

In order to extend the results of Ref.~\cite{EntroAccum} for the multiphoton case, 
one needs to present the protocol scheme in terms of an entangled state in such a way
that the projection onto the single-photon subspace followed by the normalization again gives state~(\ref{EqEPR}) and observable $\{p_b\ket u_b\bra u\}$. 
State (\ref{EqEPR}) cannot be directly generalized to the multiphoton case, since
\begin{equation*}
\begin{split}
	\frac1{\sqrt2}(\ket{j,0}_z\otimes\ket{j,0}_z+
	\ket{0,j}_z\otimes\ket{0,j}_z)
	\neq \\
	\neq\frac1{\sqrt2}(\ket{j,0}_x\otimes\ket{j,0}_x+
	\ket{0,j}_x\otimes\ket{0,j}_x)
\end{split}
\end{equation*}
if $j\geq2$.

Let us extend the space  $\mathscr H_{\widetilde A}$: 
Let $\mathscr H_{\widetilde A}=\mathbb C^2\oplus\mathbb C^4$. 
We will still use $\{\ket0_z,\ket1_z\}$ and $\{\ket0_x,\ket1_x\}$ as the bases for the space $\mathbb C^2$: 
Each element of the basis marks the sender's bit.
Take new vectors $\{\ket{0z},\ket{1z},\ket{0x},\ket{1x}\}$ as the basis for the space $\mathbb C^4$. 
Each vector marks the bit and the basis (in $\mathbb C^2$), which is used by the sender. 
An equivalent entanglement-based representation of the protocol can be constructed using the states 
\begin{equation}\label{EqEPRmulti}
\begin{split}
	\rho=
	\mu e^{-\mu}\ket\Phi\bra\Phi
	+
	\sum_{u\in\{0,1\}}\sum_{b\in\{z,x\}}
	p_be^{-\mu}\ket{ub}\bra{ub}\otimes\\
	\otimes
	\left(
	\ket{\rm vac}\bra{\rm vac}
	+
	\sum_{j=2}^\infty
	\frac{\mu^j}{j!}
	\ket{\psi^b_{ju}}\bra{\psi^b_{ju}}
	\right).
\end{split}
\end{equation}
and the observable $\Pi$,
where $\Pi_{ub}=p_b\ket u_b\bra u+\ket{ub}\bra{ub}$. 
The projection onto the single-photon component [see (\ref{EqQND})] followed by the normalization to unity turns the state $\rho$ into the projector on $\ket\Phi$.  
Then, only the second term in expression $\Pi_{ub}$ remains relevant, 
so this reproduces the previous observable.
In state (\ref{EqEPRmulti}), entanglement takes place only in the single-photon component. 
There is no need to build a state in which entanglement is also contained in multiphoton pulses, since multiphoton pulses are considered insecure in any case, and there is no need to make them secure.

Thus, a combination of Theorem~\ref{ThDecoy}, which reduces the estimation of adversary's ignorance in terms of the von Neumann entropy to the estimation of the corresponding entropy for states with single-photon pulses,
and the estimate of the latter in Ref.~\cite{EntroAccum} allows using the entropy accumulation technique for the BB84 protocol 
with the source of coherent states concluding that the Devetak--Winter formula gives the maximally achievable key rate without assumptions about the class of adversary attacks.

\section{Decoy state method}\label{sec:SecDecoy}

As follows from Eq.~(\ref{EqReduc1}), the problem of obtaining achievable key rates \ref{EqDW}) in the case of arbitrary collective attacks in the space $\mathscr F(\mathbb C^2)$ splits into two subproblems: 
The estimation of the factor $Q_1^{{\rm s}z}/Q^{{\rm s}z}$, i.e., the fraction of position in the sifted key obtained from the single-photon pulses,
and the estimation of the adversary's ignorance $H(\overline A|E)_{(1)}$ about a single bit of this part of the key.

The expression for $H(\overline A|E)_{(1)}$ is well known~\cite{ShorPreskill,Koashi,EntroAccum} and has the form 
\begin{equation}\label{EqHAE1}
	H(\overline A|E)_{(1)}=1-h(e^x_1),
\end{equation}
where $e^x_1$ is the probability of bit error in the single-photon states conditioned on the use of the $x$ basis by both parties.
Therefore, the expression for the maximally achievable key rate takes the form 
\begin{equation}\label{EqRate_GLLP}
	R=\frac{Q_1^{{\rm s}z}}{Q^{{\rm s}z}}[1-h(e^x_1)]-h(E^{{\rm s}z}).
\end{equation}
\begin{remark}\label{RemQber}
Keep in mind that, in the described variant of the BB84 protocol, the basis $x$ is rarely used.
But the probability $e^x_1$ is estimated only for those few positions in which both sides have chosen the $x$ basis.
Using Eq.~(\ref{EqRate_GLLP}), 
let us comment on the threshold error rate at which key distribution is possible, i.e. $R>0$. 
For simplicity, we assume that all the states are single photon, i.e., $Q_1^{{\rm s}z}=Q^{{\rm s}z}$. 
If the error rate does not depend on the basis, i.e. $e^x_1=E^{{\rm s}z}=e$, then the key distribution is possible if $1-2h(e)>0$, i.e. if  $e<e_{\rm crit}\approx0.11$. 
However, in a theoretically possible situation of various $e^x_1$ and $E^{{\rm s}z}$, the statement on the critical error rate of 11\%, generally speaking, is not true anymore.
For example, if $E^{{\rm s}z}=0$, then the critical error rate for $e^x_1$ increases to 50\%.
By analogy, if $e^x_1=0$, then the critical error rate for $E^{{\rm s}z}$ increases to 50\% as well.
\end{remark}

The decoy state method allows estimating the factor $Q_1^{{\rm s}z}/Q^{{\rm s}z}$ as well as the value of $e^x_1$ in an efficient manner. 
The general idea is that, in some positions randomly selected and not known in advance, the sender transmits not the signal pulses (i.e., used to form a key) with an intensity $\mu$  [see Eq.~(\ref{EqCohState})], 
but decoy pulses with lower intensities. 
The positions in which the decoy pulses are used do not participate in the key formation.
After the end of the transfer of quantum states, at the stage of announcements, the sender announces the intensities used in each message.
Based on this information, the legitimate parties calculate the statistics for each intensity and then compare the results for states of different intensities.
Informally speaking, the idea of the method is that the adversary measuring the number of photons does not know the intensity of a given pulse. 
Therefore, the adversary does not have the ability to deal differently with the signal and decoy states with the same number of photons (which, on the contrary, the adversary is supposed to know).
If the adversary blocks all single-photon components, they block almost all low-intensity decoy states that would be noticeable on the receiver side.
Formally, this leads to a zero key rate, which corresponds to the detection of eavesdropping.

The most common decoy state method, which we will describe, is based on one signal and two decoy states.
We use notations that are similar to the ones from Refs.~\cite{MaLo2005} and~\cite{Ma2017}. 
We also assume that the efficiencies and dark count probabilities of the detectors are the same.
In the opposite case, Eq.~(\ref{EqHAE1}) should be corrected; see Ref.~\cite{Bochkov2019}, 
in which an adaptation of the decoy state method for the case of detection-efficiency mismatch 
is discussed in detail.
If the efficiencies and dark count probabilities are identical, index $z$ of the values $Q_1^{{\rm s}z}$ and $Q^{{\rm s}z}$ can be removed: Detection probabilities do not depend on the basis.

Let $Y_i$ be the probability that one of the detectors  on the receiving side clicks provided that, on the sender's side, the transmitted state contained $i$ photons.
This probability does not depend on the intensity of the pulse (it depends only on the number of photons in it, i.e. $i$), but it can depend both on the attenuation in the channel and on the actions of the adversary,
which can block some of the states or, on the contrary, send them to the receiving side without attenuation.
Note that $Y_0$ is nonzero and is equal to the probability of the dark count in at least one of the detectors.

$Q_i^{v}$ is the probability that the transmitted state contains $i$ photons, 
and one of the detectors on the receiving side clicks provided that the sender used a type $v\in\{{\rm s}, {\rm d}_1, {\rm d}_2\}$ state (signal or one of two decoys).
Let us denote the intensity of the signal pulse $\mu_{\rm s}=\mu$. 
The intensities of decoy pulses are $\mu_{{\rm d}_1}=\nu_1$, $\mu_{{\rm d}_1}=\nu_2$, and we require that
\begin{equation}\label{EqCondInt}
	0\leq\nu_2<\nu_1,\quad \nu_1+\nu_2<\mu.
\end{equation}
Since the probability that the pulse contains exactly $i$ photons (under the condition of the pulse intensity $\mu$), equal to $e^{-\mu}\mu^i/i!$, then
\begin{equation}\label{EqQiv}
	Q_i^{v}=e^{-\mu_v}\frac{\mu_v^i}{i!}Y_i.
\end{equation}
The probability of a click in at least one of the detectors on the receiving side provided that the pulse is of a $v$ type is
\begin{equation}\label{EqQv}
	Q^v=\sum_{i=0}^\infty Q_i^v=\sum_{i=0}^\infty e^{-\mu_v}\frac{\mu_v^i}{i!}Y_i.
\end{equation}

Quantities $Q^v$ become known to the legitimate parties at the stage of announcements (more precisely, their estimates, but in the limit of an infinite number of messages, they coincide with the true values).
They can be used to estimate the unknown quantity $Y_1$ and, consequently, quantity $Q_1^{\rm s}$ that appears in Eq.~(\ref{EqReduc1}).

Frst, from the chain of inequalities 
\begin{equation*}
\begin{split}
	&\nu_1 Q^{{\rm d}_2}e^{\nu_2}-\nu_2 Q^{{\rm d}_1}e^{\nu_1} = \\
	&(\nu_1 - \nu_2)Y_0 - \nu_1\nu_2\sum_{i = 2}^{\infty}  (\nu_1^{i - 1} - \nu_2^{i - 1})\frac{Y_i}{i!} \leq (\nu_1 - \nu_2)Y_0,
\end{split}
\end{equation*}
we obtain the lower bound for $Y_0$:
\begin{equation}\label{EqY0L}
	Y_0 \geq Y_0^{\rm L}=\max
	\left\lbrace
	\frac{\nu_1 Q^{{\rm d}_2}e^{\nu_2}-\nu_2 Q^{{\rm d}_1}e^{\nu_1}}{\nu_1-\nu_2},0
	\right\rbrace.
\end{equation}
This inequality gives an estimate pf the probability of dark counts in at least one of the detectors of the receiving side,
which is supposed to be unknown to legitimate users in advance and may be controlled by the adversary.

The estimate for $Y_1$ comes from the chain of inequalities 
\begin{equation}
\begin{split}
	&Q^{{\rm d}_1} e^{\nu_1} - Q^{{\rm d}_2} e^{\nu_2} 
	= (\nu_1 - \nu_2)Y_1 + \sum_{i = 2}^{\infty}(\nu_1^i - \nu_2^i)\frac{Y_i}{i!} \leq\\
	&\leq (\nu_1 - \nu_2)Y_1 + \frac{\nu_1^2-\nu_2^2}{\mu^2}\sum_{i = 2}^{\infty}\frac{\mu^i}{i!}Y_i = \\
	&= (\nu_1 - \nu_2)Y_1 + \frac{\nu_1^2-\nu_2^2}{\mu^2}(Q^{\rm s} e^{\mu} - Y_0 - Y_1\mu) \leq \\
	&\leq (\nu_1 - \nu_2)Y_1 + \frac{\nu_1^2-\nu_2^2}{\mu^2}(Q^{\rm s} e^{\mu} - Y_0^{\rm L} - Y_1\mu).
\end{split}
\end{equation}
The first inequality of this chain uses conditions~(\ref{EqCondInt}) as well as the inequality $a^i - b^i \leq a^2 - b^2$ if $0 < a + b < 1$ and $i > 2$.  
Hence,
\begin{equation}\label{EqY1L}
\begin{split}
	Y_1 &\geq Y_1^{\rm L}=\frac{\mu}{\mu(\nu_1-\nu_2)-(\nu_1^2-\nu_2^2)}\times \\
	\times
	&\left[
	Q^{{\rm d}_1}e^{\nu_1}-Q^{{\rm d}_2}e^{\nu_2}
	-\frac{\nu_1^2-\nu_2^2}{\mu^2}(Q^{\rm s}e^\mu-Y_0^{\rm L})
	\right].
\end{split}
\end{equation}
and
\begin{equation}\label{EqQ1L}
	Q_1^{\rm s} \geq Q_1^{{\rm sL}}=\mu e^{-\mu}Y_1^{\rm L}.
\end{equation}

Then, we denote by $e^x_i$ the probability of bit error in the $i$-photon state provided that both parties used the $x$ basis
(not to be confused with the base of natural logarithm $e$ and the exponents $e^{\nu_1}$ and $e^{\nu_2}$ below). 
Let us estimate the probability of $e^x_1$ participating in Eq.~(\ref{EqHAE1}). 
This probability can include both imperfections of the devices and eavesdropping.
The error probability $E^{vx}$, provided that the pulse is of a $v$ type and the $x$ basis is used by both parties, is
\begin{equation*}
	E^{vx}Q^{v}=\sum_{i=0}^\infty e^x_i Y_i\frac{\mu_v^i}{i!}e^{-\mu}.
\end{equation*}
Then, from the inequality 
\begin{equation*}
\begin{split}
	&E^{{\rm d}_1x}Q^{{\rm d}_1}e^{\nu_1} - 
	E^{{\rm d}_2x}Q^{{\rm d}_2}e^{\nu_2} = \\
	&=
	e_1^x(\nu_1-\nu_2)Y_1 + \sum_{i = 2}^{\infty}e^x_i(\nu_1^i - \nu_2^i)\frac{Y_i}{i!} \geq\\
	&\geq e^x_1(\nu_1-\nu_2)Y_1 \geq e^x_1(\nu_1-\nu_2)Y_1^{\rm L},
\end{split}
\end{equation*}
we obtain an upper bound on $e^x_1$:
\begin{equation}\label{Eqe1U}
	e_1^x \leq e_1^{x\rm U}=\frac{E^{{\rm d}_1x}Q^{{\rm d}_1}e^{\nu_1} - E^{{\rm d}_2x}Q^{{\rm d}_2}e^{\nu_2}}{(\nu_1-\nu_2)Y_1^{\rm L}}.
\end{equation}

The substitution of the obtained estimates into Eqs.~(\ref{EqHAE1}) and (\ref{EqReduc1}) yields the formula for the achievable secret key rate:
\begin{equation}\label{EqRate}
	R=\frac{Q^{\rm sL}_1}{Q^{\rm s}}[1-h(e^{x\rm U}_1)]-h(E^{{\rm s}z}),
\end{equation}

In the absence of eavesdropping, 
the exact values of all quantities that participate in the definition of the key rate take the form 
[see Eq.~(5)--(11) in Ref.~\cite{MaLo2005}]
\begin{align}
	&Q^{\rm s}=p_d+1-e^{-\eta\mu T(L)},\label{EqQmu}\\
	&Y_1=p_d+\eta T(L),\label{EqY1}\\
	&Q_1^{\rm s}=[p_d+\eta T(L)]\mu e^{-\mu}\label{EqQ1},\\
	& E^{{\rm s}z}=p_d/2\label{EqErrMu},\\
	&e^x_1=\frac{p_d}{2Y_1},\label{Eqe1}
\end{align}
where $p_d$ is the probability of a dark count of at least one of the detectors,
$T(L)=10^{-\delta L/10}$ is the transmission coefficient in a channel of length $L$, 
$\delta$ is the specific loss factor,
$\eta$ is the quantum efficiency each of single-photon detector.
We assume that errors occur only due to dark counts, and the optical part is perfectly tuned.

It is possible to consider the dependence of the key length on the length of the communication line with realistic parameters on the receiving side: 
The intensity of the signal and decoy states $\mu=0.5$, $\nu_1=0.01$, $\nu_2=0.001$,  $p_d=10^{-6}$, $\eta=0.1$, $\delta=0.2$ dB/km. 
In Fig.~\ref{fig1}, the plot of $R$, which is calculated using Eqs.~(\ref{EqRate}) and (\ref{EqRate_GLLP}), as a function of the communication channel length is presented. 
In both cases, the observed statistics of $Q^{\mu}$, $Q^{\nu_1}$, $Q^{\nu_2}$, $E^{\mu}$, $E^{\nu_1}$, and $E^{\nu_2}$ are calculated using Eqs.~(\ref{EqQmu}) and (\ref{EqErrMu}), i.e. for the case of the absence of eavesdropping.

\begin{figure}
	\centering
	\includegraphics[width=1\linewidth]{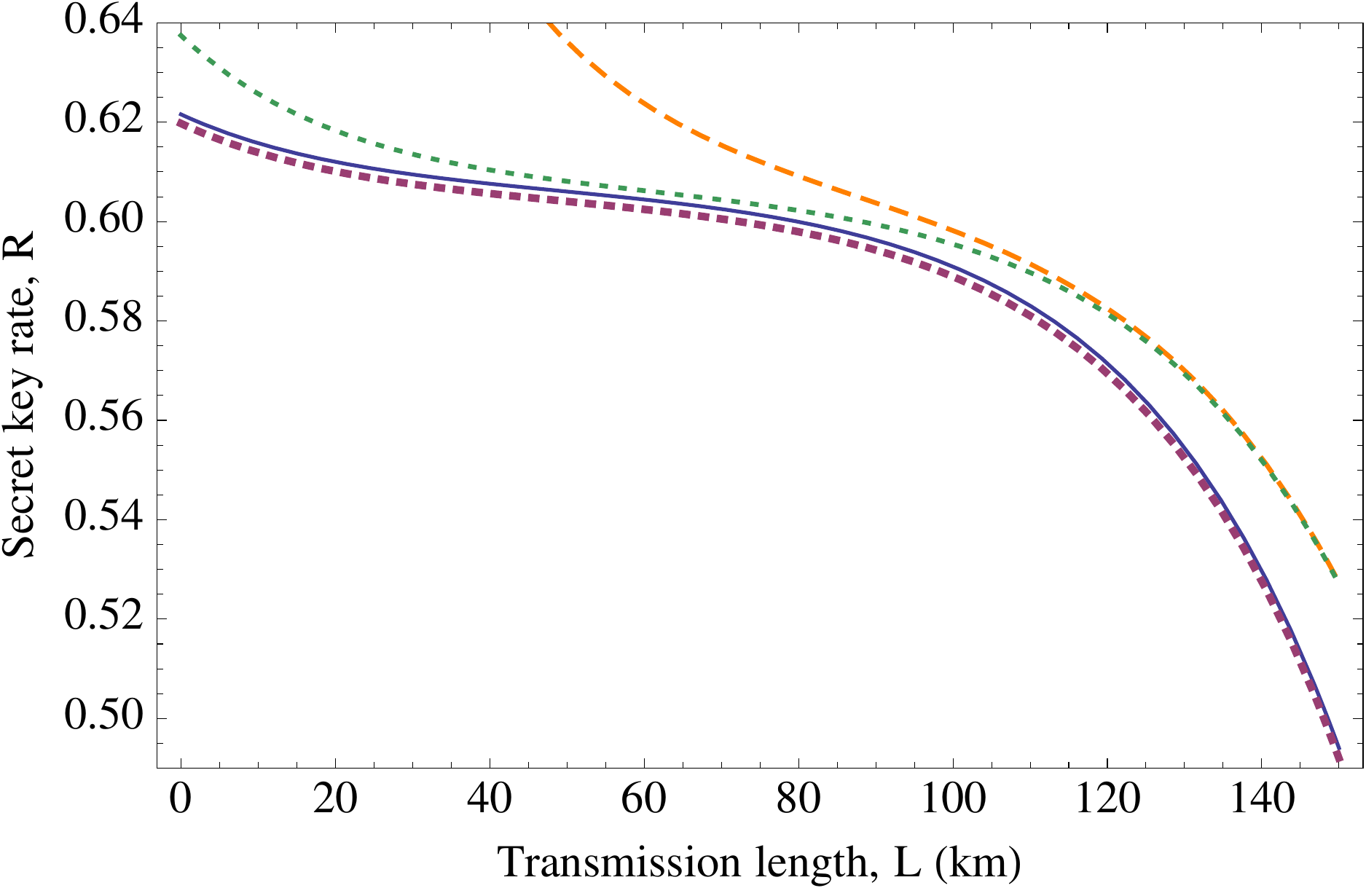}
	\caption{Secret key rate depending on the communication channel length.
	Solid blue line: Calculation for a photon-number splitting attack with a coherent attack on one-photon states by Eq.~(\ref{EqRate_GLLP}). 
	Solid dotted purple line: Calculation for the same attack using Eq.~(\ref{EqRate}). 
	In this formula, instead of, generally speaking, unknown values of $Q_1^\mu$ and $e_1$, their lower and upper estimations are used, respectively. 
	(formulae (\ref{EqQ1L}) and (\ref{Eqe1U}) with formulae (\ref{EqQmu}) and (\ref{EqErrMu}) for values of the observed statistics of registrations. 
	Thin dotted green line: Calculation for a beam splitting attack [Eq.~(\ref{EqRateBS}) with $t=\eta T(L)$].
	Thin dashed orange line: Calculation for a beam splitting attack, in which the adversary does not change the efficiency of detectors on the receiver side [Eq.~(\ref{EqRateBS}) with $t=T(L)$].
	Parameters: Intensities of signal and decoy states are $\mu=0.5$, $\nu_1=0.01$, $\nu_2=0.001$, dark count probability is $p_d=10^{-6}$,
	quantum efficiency of each single-photon detector is $\eta=0.1$,
	specific loss factor in the communication channel is $\delta=0.2$ dB/km.}
\label{fig1}
\end{figure}

In Ref.~\cite{MaLo2005} it is shown that at $\nu_1,\nu_2\to0$,
and without eavesdropping 
(or when the eavesdropping does not change the registration statistics), 
the estimates $Q^{\mu,\rm L}_1$ and $e_1^{\rm U}$ become exact, 
i.e. they tend to the values of $Q_1^\mu$ and $e_1$, 
which are defined by formulae (\ref{EqQ1}) and (\ref{Eqe1}). 
Then, in the asymptotic case, the key rate is defined by Eq.~(\ref{EqRate_GLLP}), 
where $Q^{\mu}$, $Q_1^\mu$, and $e_1$ are defined by formulae~(\ref{EqQmu})--(\ref{Eqe1}).

\begin{remark}
It is important to note that estimates~(\ref{EqY1L})--(\ref{Eqe1U}) do not rely on knowledge of the transmission coefficient of the channel $T(L)$,
dark count probability $p_d$, detector efficiencies $\eta$, or other parameters of the optical and detection system of the receiving side.
In the derivation of these estimates, it is not assumed that these coefficients are the same for all pulses.
Thus, in fact, the decoy state method provides security even when the adversary can change these parameters. 
This is explained in detail, for example, in Ref.~\cite{Bochkov2019}.
Thus, we assume that the adversary can reduce natural losses and natural errors in the communication channel and even in the detectors up to creating a communication line and detectors without losses and errors.
This gives the adversary certain freedom in creating the desired statistics of registration of states with different intensities on the receiving side.
\end{remark}

\section{Beam splitting attack}\label{sec:SecBS}

As was noted in Sec.~\ref{sec:SecReduc}, at the practical level of losses in the channel, estimates~(\ref{EqReduc}) and, therefore, (\ref{EqRate_GLLP}) become exact 
when the adversary performs a photon-number splitting attack. 
However, this attack has not yet been implemented in practice.
Let us compare this attack with the simpler beam splitting attack, which is also often considered in quantum cryptography~\cite{Gisin2002,BB92,Dusek,Felix}. 

A simple version of the attack consists of the fact that the adversary, with the use of a beam splitter, takes away the part of the signal (a coherent state) that is absorbed in the channel as a result of natural losses.
The adversary transmits the rest of the signal to the receiver via an ideal channel, i.e., a lossless channel.
As a result, this attack does not change the statistics of the receiver state registrations 
in comparison with the natural conditions (in the absence of an adversary).

From the formal point of view, the attack is described by the isometry
$V_{\rm BS}:\mathscr F(\mathbb C^2)\to\mathscr F(\mathbb C^2)\otimes\mathscr F(\mathbb C^2)$, 
for the determination of which it is enough to specify its action on all possible coherent states~(\ref{EqCohState}):
\begin{equation}\label{EqBS}
	V_{\rm BS}\ket{\alpha,u}_b=\ket{\sqrt t\,\alpha,u}_b\otimes 
	\ket{\sqrt{1-t}\,\alpha,u}_b,
\end{equation}
where $t$ and $1-t$ are transmittance and refractive indices, respectively.
The first space in the tensor product of $V_{\rm BS}$ output is interpreted as the receiver's space,
while the second is the adversary's space.
The composition of the map
\begin{equation}\label{EqBSTrans}
	\rho^z_{\mu u}\mapsto V_{\rm BS}\,\rho^z_{\mu u}V_{\rm BS}^\dag,
\end{equation}
describing the eavesdropping, with a map describing the measurement on the receiver side gives channel $\Upsilon_0$ (see Sec.~\ref{sec:SecReduc}), which corresponds to this attack.
In any case, since, after map (\ref{EqBSTrans}), the state of the adversary does not undergo changes anymore, 
map (\ref{EqBSTrans}) is enough to determine the state $\rho_{\overline AE}$ and,
therefore, define entropy $H(\overline A|E)$ in Eq.~(\ref{EqDW}):
\begin{equation*}
\begin{split}
	&\rho_{\overline AE}=\frac12\sum_{u=0}^1\ket u_{\overline A}\bra u
	\otimes\Tr_B(V_{\rm BS}\,\rho^z_{\mu u}V_{\rm BS}^\dag)
	=\\
	&=\frac{e^{-\mu_E}}2
	\sum_{u=0}^1\ket u_{\overline A}\bra u \otimes \\
	&\otimes
	\left[
	\ket{\rm vac}\bra{\rm vac}
	+\sum_{j=1}^\infty
	\frac{\mu_E^j}{j!}
	\ket{\psi_{j\mu_E}^z}_E\bra{\psi_{j\mu_E}^z}
	\right]
	=\\
	&=
	\frac{e^{-\mu_E}}2I_{\overline{A}}\otimes
	\ket{\rm vac}\bra{\rm vac}+
	\frac{e^{-\mu_E}}2
	\sum_{u=0}^1\ket u_{\overline A}\bra u\otimes \\
	&\otimes
	\sum_{j=1}^\infty
	\frac{\mu_E^j}{j!}
	\ket{\psi_{j\mu_E}^z}_E\bra{\psi_{j\mu_E}^z},
\end{split}
\end{equation*}
where $\mu_E=(1-t)\mu$, $u\in\{0,1\}$. 
We see that the adversary has zero information about the key bit (accordingly, the ignorance is equal to one) if the vacuum component is actualized (with probability $e^{-\mu_E}$). 
Otherwise, the adversary has complete information about the key bit.
Since state~(\ref{EqBS}) has the form of a tensor product, the realization of the adversary's vacuum component does not depend on the photon registration event on the receiver's side.
Therefore, the probability of the vacuum component being conditioned on a click of a receiver's detector is equal to the unconditional probability of this event.
Therefore, $H(\overline A|E)=e^{-(1-t)\mu}$ and
\begin{equation}\label{EqRateBS}
	R_{\rm BS}=e^{-(1-t)\mu}-h(E^{{\rm s}z}).
\end{equation}

In the absence of eavesdropping, the actual transmission coefficient of the entire optical system 
(communication line and detectors) is $\eta T(L)$.
If we assume that the portion of the signal absorbed in the channel and on the detectors goes entirely to the adversary, then $t=\eta T(L)$ should be set.
This is usually explained as follows: 
We assume that the adversary can replace the communication line and detectors with ideal ones (i.e., without losses) and, due to this, take away a share of the signal $1-\eta T(L)$ 
reproducing the natural level of losses.
If we assume that the adversary can replace the communication line with an ideal one, but has no power over the efficiency of the detectors $\eta$, then we can put $t=T(L)$, 
which reduces the adversary's information about the key and allows the legitimate parties 
to generate the key at a higher rate, under this assumption about the adversary.

In Review~\cite{Gisin2002}, the feasibility of the assumption about the ability of the adversary to replace the channel with a noiseless one is discussed.
It is noted that all existing solutions have fundamental limitations on the minimal loss factor.
There are hypothetical solutions that are theoretically capable of providing a lossless channel, but this review concludes that they are unrealistic in any foreseeable future.
However, quantum cryptography often assumes unrealistically large adversary capabilities in order to fundamentally resolve the issue of secure key distribution.
So, for example, even if the adversary cannot create a lossless communication line, 
they can create a communication line with a lower level of losses. 
It can also affect the efficiency of detectors within certain limits.
In order not to discuss the issues of what these limits may be, one can make the weakest assumption about the complete control of the adversary over the transmission coefficient of the line and the effectiveness of the detectors.
But, let us recall, here we are considering the case of identical detector efficiencies, i.e., we are assuming that the adversary can only change the efficiency of all detectors by the same amount.

Plots of $R_{\rm BS}$ for $t=\eta T(L)$ and $t=T(L)$ are presented in Fig.~\ref{fig1}.
By comparing them with the calculated achievable rates (\ref{EqRate_GLLP}) and (\ref{EqRate}), one can conclude that $R_{\rm BS}$ is higher for the calculated achievable rates.
This was expected, since a beam splitting attack is a particular type of attack, while formulas (\ref{EqRate_GLLP}) and (\ref{EqRate}) are derived for the general case.
In view of the importance of the beam splitting attack, in the next section, we consider why the beam splitting attack is less effective than the photon-number splitting attack.

By comparing plots $R_{\rm BS}$ for $t=\eta T(L)$ and $t=T(L)$, we see that, 
if the length is small, so that the main contribution to the losses comes from the detectors, 
then the possibility of the adversary diverting a larger portion of the signal by increasing the efficiency of the detectors greatly reduces the achievable key rate and brings it closer to the achievable key rate for the photon-number splitting attack.
At communication line lengths from 50 km, the main contribution to the losses is no longer from the detectors, but from the communication channel; 
therefore, the curves for these two cases approach each other and practically coincide at lengths from 120 km.

Note that, in recent works~\cite{KK,KKFK,AKP,KNKF}, generalizations of the beam splitting attack for the B92,
COW, and DPS protocols are considered. 
In these generalizations, the adversary can change the intensity of states, and their further actions depend on
whether they have managed to extract information from the part of the state diverted by the beam splitter.

\section{Comparison of photon-number splitting and beam splitting attacks}\label{sec:SecProof}

Beam splitting and photon-number splitting attacks are very important in quantum cryptography, so, in this section, we explain why the first attack is less effective than the second one.
To do this, we first prove the inequality
\begin{equation}\label{EqIneq}
	R<R_{\rm BS},
\end{equation}
where $R$  is calculated using formula~(\ref{EqRate_GLLP}) and $R_{\rm BS}$ is calculated using formula (\ref{EqRateBS}) for $t=\eta T(L)$ (i.e. with the weakest assumption about the adversary).
If we replace inequality  (\ref{EqIneq}) with the nonstrict one, 
then it follows directly from formula~(\ref{EqRate_GLLP}), because formula~(\ref{EqRate_GLLP}) is derived for a general attack. 
Let us prove a strict inequality, i.e., that the beam splitting attack is always suboptimal. 
We also recall that the photon-number splitting attack is optimal, 
i.e., key rate (\ref{EqRate_GLLP}) is maximally achievable.

Let us prove (\ref{EqIneq}) for the case $\mu\leq1$. 
The case $\mu>1$ can also be considered, but in practice, $\mu\leq1$, so, for simplicity, we restrict ourselves to considering only this case.
We also assume that the quantity $\eta T(L)$ is strictly positive (otherwise, not only the key distribution, 
but also any communication is impossible).
Recall also that  $\eta T(L)\leq1$.

Since the term $h(E^{{\rm s}z})$ in both formulae~(\ref{EqRate_GLLP}) and (\ref{EqRate}) is the same, the proof of Ineq.~(\ref{EqIneq}) is reduced to the proof of the inequality
\begin{equation}\label{EqIneq00}
	\frac{Q_1^\mu}{Q^{\mu}}[1-h(e^x_1)]<e^{-(1-t)\mu}.
\end{equation}
Let us prove a stricter inequality,
\begin{equation}\label{EqIneq0}
	\frac{Q_1^\mu}{Q^{\mu}}<e^{-(1-t)\mu}.
\end{equation}
Using (\ref{EqQmu})--(\ref{EqQ1}), for the left-hand side, we have
\begin{equation}
	\frac{Q_1^\mu}{Q^{\mu}}=
	\frac{\mu e^{-\mu}(p_d+t)}{p_d+1-e^{-\mu t}}.
\end{equation}
Let us consider the fraction
\begin{equation}
	\frac{p_d+t}{p_d+1-e^{-\mu t}}.
\end{equation}
In view of the inequality $1-e^{-x}<x$ for $x>0$ and, therefore, $1-e^{-\mu t}<\mu t\leq t$, 
one can conclude that the fraction decreases with increasing $p_d$, so that
\begin{equation}
	\frac{\mu e^{-\mu}(p_d+t)}{p_d+1-e^{-\mu t}}
	\leq 
	\frac{\mu t e^{-\mu}}{1-e^{-\mu t}}.
\end{equation}
Thus, it is necessary to prove the inequality
\begin{equation}\label{EqIneq2}
	\frac{\mu t e^{-\mu}}{1-e^{-\mu t}}<e^{-\mu(1-t)},
\end{equation}
or 
\begin{equation}\label{EqIneq3}
	\frac{\mu t}{e^{\mu t}-1}<1.
\end{equation}
This inequality obviously follows from $e^{x}-1>x$ for $x>0$.
Ineq.~~(\ref{EqIneq0}) and, therefore, Ineq.~(\ref{EqIneq}) have been proved. 

Let us analyze the reasons why Ineq.~(\ref{EqIneq}) is strict, so that the beam splitting attack is not optimal.
First, in the proof of Ineq.~(\ref{EqIneq}) we have used the inequality $h(e_1^x)>0$ for $e_1^x>0$. 
This means that the beam splitting attack does not attack single-photon transmissions
at the cost of introducing noise, in contrast to the photon-number splitting attack.
This is the first reason why the beam splitting attack is non-optimal.

Second, for~(\ref{EqIneq0}) arises from the following effect.
Since the sent state is a mixture of Fock states, we can think about a certain number of photons in a pulse, even if this observable is not actually measured by anyone.
One can imagine another participant who, before the actions of the adversary, measures the number of photons in the pulse (without changing the state) and then observes the adversary's actions with the knowledge of the number of photons.
They then observe that, in the case of the photon-number splitting attack, 
the adversary recognizes the information encoded in multiphoton pulses with certainty,
while in the case of the beam splitting attack, with a certain probability, 
they can transmit all photons to the receiver without leaving a single photon for themselves or, 
conversely, they take all the photons for themselves
without passing a single one to the receiver. 
The last alternative leads to the no registration event (unless a dark count occurs) on the receiver's side, so this position is not included in the sifted key.
In this case, the fraction of multiphoton (insecure) states in the sifted key decreases. 
In contrast, in the photon-number splitting attack, only the single-photon transmissions are blocked.
Probabilistic processing of multiphoton pulses in the beam splitting attack, instead of optimal deterministic processing in the photon-number splitting attack, is the reason for Ineq.~(\ref{EqIneq0}). 
Moreover, if, after measuring the number of photons, 
the adversary performs a probabilistic processing of the multiphoton pulses, 
then they can reproduce (simulate) the beam splitting attack.

In Ref.~\cite{Molotkov2019}, it is noted that, in the beam splitting attack, the adversary takes away almost all quantum states.
However, only those positions in which the receiver also registers a photon participate in the key generation.
Therefore, the withdrawal of the adversary of almost all quantum states leads to losses but does not necessarily lead to a large amount of knowledge of the adversary about the sifted key.

From Fig.~\ref{fig1}, we see that, 
at the considered parameters and the length of the communication line up to approximately 140~km, 
the possibilities of these two attacks practically coincide:
The achievable secret key rate for the case of the beam splitting attack is only slightly higher than the achievable secret key rate for the case of the photon-number splitting attack.
At communication channel lengths over 140 km, the advantages of the photon number splitting attack become significant, and the curves diverge.
Over a length of approximately 200 km, the achievable rate in the case of the photon-number splitting attack drops to zero, while the secret key can still be generated in the case of the beam splitting attack.

From the above reasoning, it is possible once again (see the end of Sec.~\ref{sec:SecReduc}) to easily deduce the optimality of the photon-number splitting attack.
If it is possible to measure the number of photons without spoiling the state, it is optimal to do this. 
Next, one should take actions that are optimal for a given known number of photons.
Obviously, if the state is multiphoton, then it is optimal to remove one photon from it (more is possible, but one photon also gives complete information about the key bit).
If the transmission is single photon, it is optimal to carry out the optimal attack for the single-photon pulses at the cost of introducing noise.
Since the transmitted state is a mixture of Fock states and any quantum transformation is linear,
any possible attack can be simulated by measuring the number of photons and applying one probabilistic processing of a pulse or another depending on the number of photons in it.
In fact, Theorem~\ref{ThDecoy} is a formalization of these discussions. 

In concluding this section, let us mention the unambiguous state discrimination attack (USD attack). 
For the BB84 protocol, it is discussed in Ref.~\cite{Dusek}.\,It is emphasized that, in the case of coherent states, with a randomized phase, the photon-number splitting attack (without attacking the single-photon pulses) 
is an optimal discrimination of states with an unambiguous outcome.

\section{Polarization and phase encoding}\label{sec:SecPhaseCoding}

Some researchers cast doubts on the applicability of the decoy state method in the case of the widely used 
phase encoding instead of polarization encoding~\cite{KulikMol}.
To eliminate these doubts, in this section, we show that these two encoding methods are completely equivalent; 
therefore, the decoy state method is applicable to phase encoding in the same way as to the polarization one.

Note that the previous discussion did not rely on what kind of encoding is used, although we have referred to the polarization encoding, for example.

In the polarization encoding, creation operators $a^\dag_{z0}$, $a^\dag_{z1}$, $a^\dag_{x0}$, and $a^\dag_{x1}$ in Eqs.~(\ref{Eqj0j1z})--(\ref{Eqax}) can be, 
for example, photon creation operators with horizontal, vertical, diagonal, and antidiagonal polarizations: $a^\dag_{H}$, $a^\dag_{V}$, $a^\dag_{D}$, and $a^\dag_{A}$, respectively. 
The transmitted states can then be rewritten in the form
\begin{equation}\label{EqPolCoding}
\begin{split}
	\ket{\alpha,0}_z&=\ket{\alpha}_{z0}\ket{0}_{z1}\equiv\ket{\alpha}_{H}\ket{0}_{V},\\
	\ket{\alpha,1}_z&=\ket{0}_{z0}\ket{\alpha}_{z1}\equiv\ket{0}_{H}\ket{\alpha}_{V},\\
	\ket{\alpha,0}_x&=\ket{\alpha}_{x0}\ket{0}_{x1}\equiv\ket{\alpha}_{D}\ket{0}_{A}
	=
	\left|\frac\alpha{\sqrt2}\right\rangle_{H}
	\left|\frac\alpha{\sqrt2}\right\rangle_{V},\\
	\ket{\alpha,1}_x&=\ket{0}_{x0}\ket{\alpha}_{x1}\equiv\ket{0}_{D}\ket{\alpha}_{A}
	=
	\left|\frac\alpha{\sqrt2}\right\rangle_{H}
	\left|-\frac\alpha{\sqrt2}\right\rangle_{V},
\end{split}
\end{equation}
where $\ket\alpha=e^{-\mu/2}\sum_{j=1}^\infty(\alpha^j/\sqrt{j!})\ket j$ and $\ket j$, $j=0,1,\ldots$ are a coherent state and a state with a definite number of photons in the corresponding mode, respectively, $\alpha\in\mathbb C$, $\mu=|\alpha|^2$.

However, in the phase encoding, the employed states have the following form~\cite{Gisin2002,KulikMol}:
\begin{equation}\label{EqPhCoding}
	\left|(-1)^u \frac{e^{i\varphi_b}\alpha}{\sqrt2}
	\right\rangle_1
	\left|\frac{\alpha}{\sqrt2}\right\rangle_2,
\end{equation}
where $u\in\{0,1\}$ is the coding bit, $b\in\{z,x\}$ is the basis, $\varphi_z=0$, $\varphi_x=\pi/2$, and modes 1 and 2 correspond to two time windows. 
That is, each signal consists of two phase-matched pulses.
The phase in each pair of pulses is randomly selected, i.e. the adversary and the receiver `observe' state of the form 
\begin{equation}\label{EqPhCodingRND}
\begin{split}
	\frac1{2\pi}\int_0^{2\pi}d\theta\,
	\big(
	\ket{(-1)^u e^{i(\varphi_b+\theta)}\sqrt{\mu/2}}_1\ket{e^{i\theta}\sqrt{\mu/2}}_2
	\big)\times\\
	\times
	\big(
	{}_2\bra{e^{i\theta}\sqrt{\mu/2}}_1\bra{(-1)^u e^{i(\varphi_b+\theta)}\sqrt{\mu/2}}
	\big).
\end{split}
\end{equation}
The interference scheme on the receiver side is designed in such a way that only the phase difference in the two windows matters.

Let us show that the states of polarization encoding can be represented in a form equivalent to expression~(\ref{EqPhCoding}). 
Consider creation operators of photons with right and left circular polarizations:
\begin{equation}\label{EqaRLHV}
\begin{split}
	a_R^\dag&=\frac{a_H^\dag-ia_V^\dag}{\sqrt2},\\
	a_L^\dag&=\frac{a_H^\dag+ia_V^\dag}{\sqrt2}.\\
\end{split}
\end{equation}
Then, states (\ref{EqPolCoding}) can be written in the from 
\begin{equation}\label{EqPolCodingRL}
\begin{split}
	\ket{\alpha,0}_z&=\left|\frac\alpha{\sqrt2}\right\rangle_{R}
	\left|\frac\alpha{\sqrt2}\right\rangle_{L},\\
	\ket{\alpha,1}_z&=\left|\frac{i\alpha}{\sqrt2}\right\rangle_{R}
	\left|-\frac{i\alpha}{\sqrt2}\right\rangle_{L},\\
	\ket{\alpha,0}_x&=\left|\frac{(1+i)\alpha}{2}\right\rangle_{R}
	\left|\frac{(1-i)\alpha}{2}\right\rangle_{L},\\
	\ket{\alpha,1}_x&=\left|\frac{(1-i)\alpha}{2}\right\rangle_{R}
	\left|\frac{(1+i)\alpha}{2}\right\rangle_{L}.
\end{split}
\end{equation}
The first state coincides with state~(\ref{EqPhCoding}) for $u=0$ and $b=z$, 
the second state coincides with the corresponding state (\ref{EqPhCoding}) when both phases are shifted by $\pi/2$ (i.e. with $\alpha$ replaced by $e^{i\pi/2}\alpha$), 
the third state is when phases are shifted by $\pi/4$, and, finally, the fourth state is when phases are shifted by $-\pi/4$.
In view of the phase randomization (\ref{EqRhoPhRnd}) and (\ref{EqPhCodingRND}), 
these phase shifts play no role. 
Application of the inverse phase shifts in expressions~(\ref{EqPhCoding}) 
and the inverse transformation with respect to transformation~(\ref{EqaRLHV})
(with a different label for the modes, since they now do not correspond to the polarizations) gives states (\ref{EqRhoPhRnd}).

Thus, polarization and phase encoding are completely mathematically equivalent: 
They correspond to different mode expansions of the same states.
That is why we can conclude that Theorem~\ref{ThDecoy} and estimates (\ref{EqQ1L}) and (\ref{Eqe1U}) do not depend on the type of encoding: 
States of both polarization and phase encoding can be expressed in form~(\ref{EqRhoPhRnd}).

This is confirmed by the results from Ref.~\cite{KulikMol}, where for the case of phase coding, the same decoy state method estimates are obtained as for the polarization coding.
But it should be noted that an equivalence takes place when the total pulse intensity in two windows in Eq.~(\ref{EqPhCoding}) coincides with the pulse intensity for polarization coding (\ref{EqPolCoding}) $\mu=|\alpha|^2$. 
The pulse intensity in each window should be, respectively, one half.
Taking into account this remark, the estimates derived in Ref.~\cite{KulikMol} completely coincide with the well-known estimates derived in Ref.~\cite{MaLo2005}.

It is important that transformations (\ref{EqaRLHV}) preserve the photon number:
\begin{equation*}
	a_H^\dag a_H+a_V^\dag a_V=a_D^\dag a_D+a_A^\dag a_A=a_R^\dag a_R+a_L^\dag a_L.
\end{equation*}
This means that estimates related to the number of single-photon pulses and the corresponding errors obtained for one mode expansion are also valid for the other expansion of the same state.

Using the described transformations, it is possible, for example, to reformulate the photon-number splitting attack, which was originally formulated for polarization encoding, for phase encoding.
Of course, the mathematical equivalence of the two coding methods does not imply the same technological complexity of the photon-number splitting attack in these two cases.
However, when calculating the key length, we are not interested in the technological complexity of certain adversary operations, 
since we assume that the adversary can carry out any transformations allowed by the mathematical apparatus.

\section{Conclusion}

In this paper, we have presented the decoy state method in quantum cryptography, focusing on issues that are usually not covered in the literature.
The first one is a formal proof that the security of the protocol with a phase-randomized coherent state source is reduced to the security of the corresponding single-photon source protocol.
This proof strongly justifies the security of the decoy state protocol against all kinds of attacks, not just the photon-number splitting attack.

The emphasis on the photon-number splitting attack in the literature is related to the fact that this attack is optimal.
However, the security proof of the decoy state method does not rely on this fact.
Notably, we have compared the photon-number splitting attack with the beam splitting attack and both analytically and numerically demonstrated the lower efficiency of the latter.

Another issue considered in the paper is the equivalence of polarization and phase encoding from the viewpoint of the security of quantum cryptography with decoy states.
Despite the fact that, technologically, these two methods of state encoding and, therefore, attacks on the corresponding implementations are significantly different, they are completely equivalent mathematically.
In the security analysis, we assume that the adversary can realize any transformations that are allowed by the mathematical apparatus;
therefore, the argument of mathematical equivalence is sufficient to justify the equal theoretical security of protocols using these two types of encoding.

Thus, in numerous theoretical papers dealing with the decoy state BB84 protocol,
as well as in experimental implementations of this protocol, including the implementation described in Ref.~\cite{Duplinskiy2018}, the estimate of the secret key length is justified and rigorously proven.

We also note that, in Ref.~\cite{Trushechkin2017}, a method of calculating estimates for the decoy state method taking into account statistical fluctuations is proposed. 
Also, in Ref.~\cite{Bochkov2019}, a generalization of the decoy state method estimates in the case of detection-efficiency mismatch is presented.
\newpage
\section*{Acknowledgements}

The research in the Secs.~\ref{sec:SecPNS}-\ref{sec:SecEntroAccum} (the main result is Theorem~\ref{ThDecoy}) was carried out within the state assignment of the Ministry of Science and Higher Education of the Russian Federation.
Research in Secs.~\ref{sec:SecBS}-\ref{sec:SecProof} 
[the main result is inequality~(\ref{EqIneq})] and \ref{sec:SecPhaseCoding}
(proof of equivalence of polarization and phase encoding and the graph in Fig.~\ref{fig1} in Sec.~\ref{sec:SecDecoy})
was completed with the support of the grant from a Russian Science Foundation (project No. 17-71-20146).

\end{document}